%% file: Main.tex
\documentclass[twocolumn]{article}

\usepackage{sansmathfonts} % Use the Palatino font
\usepackage[default,oldstyle,scale=0.9]{opensans} %% Alternatively
\usepackage{geometry}
\usepackage[T1]{fontenc} % Use 8-bit encoding that has 256 glyphs
\usepackage{microtype} % Slightly tweak font spacing for aesthetics

\usepackage[english]{babel} % Language hyphenation and typographical rules

\usepackage[small,labelfont=bf,up,textfont=it,up]{caption} % Custom captions under/above floats in tables or figures
\usepackage{booktabs} % Horizontal rules in tables

\usepackage{lettrine} % The lettrine is the first enlarged letter at the beginning of the text

\usepackage{enumitem} % Customized lists
\setlist[itemize]{noitemsep} % Make itemize lists more compact

\usepackage{abstract} % Allows abstract customization
\usepackage{titlesec} % Allows customization of titles
\renewcommand\thesection{\Roman{section}} % Roman numerals for the sections
\renewcommand\thesubsection{\roman{subsection}} % roman numerals for subsections
\titleformat{\section}[block]{\large\scshape\centering}{\thesection.}{1em}{} % Change the look of the section titles
\titleformat{\subsection}[block]{\normalsize}{\thesubsection.}{1em}{} % Change the look of the section titles

\usepackage{titling} % Customizing the title section

\usepackage{cite}
\usepackage{amsmath,amssymb,amsfonts, amsthm}
\usepackage{algorithmic}

\usepackage[dvips]{graphicx}
\DeclareGraphicsExtensions{.pdf,.jpeg,.png,.eps}
\usepackage{textcomp}
\usepackage{xcolor}
\usepackage{threeparttable}
\usepackage{color, colortbl} % the first to define new colors and the latter to actually color the table
\usepackage{algorithmic}
\usepackage{algorithm} % there is specific instruction not to use this
\usepackage{subcaption}
\usepackage{float}
 \usepackage{url}
 \usepackage{verbatim}
\usepackage{lipsum}
\definecolor{Gray}{gray}{0.9} %light gray
\definecolor{LightCyan}{rgb}{0.88,1,1} %Light cyan:
\definecolor{paleblue}{rgb}{0.69, 0.93, 0.93}
\definecolor{paleaqua}{rgb}{0.74, 0.83, 0.9}
\newtheorem{prop}{Proposition}
\newtheorem{corr}{Corollary}

\newcommand*{\transpose}{\mathsf{T}}

\pretitle{\begin{center}\Large\bfseries} % Article title formatting
\posttitle{\end{center}} % Article title closing formatting
\title{Multi-Criteria Radio Spectrum Sharing With Subspace-Based Pareto Tracing} % Article title
\author{Zachary~J.~Grey,~
        Susanna~Mosleh,~Jacob~D.~ Rezac,~ Yao~Ma,\\Jason~B.~Coder,
        and~Andrew~M.~ Dienstfrey% <-this % stops a space
\thanks{
Z. J. Grey and A. M. Dienstfrey are with the Applied and Computational Mathematics Division, Information Technology Lab, National Institute of Standards and Technology, Boulder, CO, 80305
\newline
\indent S. Mosleh, is an associate with the RF Technology Division, Communications Technology Lab, National Institute of Standards and Technology, Boulder, CO, 80305.
\newline
\indent J. D. Rezac, Y. Ma and J. B. Coder are with the RF Technology Division, Communications Technology Lab, National Institute of Standards and Technology, Boulder, CO, 80305.
\newline
\indent Part of this work was presented at the IEEE International Conference on Communications (ICC 2021) \cite{Grey2021}.
\newline
\indent This work is U.S. Government work and not protected by U.S. copyright.
}
}
\date{\today} % Leave empty to omit a date
\renewcommand{\maketitlehookd}{%
\begin{abstract}
\noindent Radio spectrum is a high-demand finite resource. To meet growing demands of data throughput for forthcoming and deployed wireless networks, new wireless technologies must operate in shared spectrum over unlicensed bands (coexist). As an example application, we consider a model of Long-Term Evolution (LTE) License-Assisted Access (LAA) with incumbent IEEE 802.11 wireless-fidelity (Wi-Fi) systems in a coexistence scenario. This scenario considers multiple LAA and Wi-Fi links sharing an unlicensed band; however, a multitude of other scenarios could be applicable to our general approach. We aim to improve coexistence by maximizing the key performance indicators (KPIs) of two networks simultaneously via dimension reduction and multi-criteria optimization. These KPIs are network throughputs as a function of medium access control (MAC) protocols and physical layer parameters. We perform an exploratory analysis of coexistence behavior by approximating active subspaces to identify low-dimensional structure in the optimization criteria, i.e., few linear combinations of parameters for simultaneously maximizing LAA and Wi-Fi throughputs. We take advantage of an aggregate low-dimensional subspace parametrized by approximate active subspaces of both throughputs to regularize a multi-criteria optimization. Additionally, a choice of two-dimensional subspace enables visualizations augmenting interpretability and explainability of the results. The visualizations and approximations suggest a predominantly convex set of KPIs over active coordinates leading to a regularized manifold of approximately Pareto optimal solutions. Subsequent geodesics lead to continuous traces through parameter space constituting non-dominated solutions in contrast to random grid search.
\end{abstract}
}

\begin{document}

% Print the title
\maketitle

%----------------------------------------------------------------------------------------
%	ARTICLE CONTENTS
%----------------------------------------------------------------------------------------

\section{Introduction}
\input{Intro}

\section{System Model and Problem Formulation} \label{sec:sys_model}
\input{Sys_Model_Simple}

\section{Pareto Tracing} \label{sec:pareto_tracing}
\input{Pareto_tracing}

\section{Active Subspaces} \label{sec:AS}
\input{AS_intro}

\section{Simulation and Results} \label{sec:Numerics}
\input{Numerics}

\section{Conclusion \& Future Work} 
We have proposed a technique to simultaneously optimize the performance of two MNOs sharing limited unlicensed spectrum resources. An exploratory analysis utilizing an example of LAA coexistence with Wi-Fi network identified a common subspace-based dimension reduction of a basic model of network behavior. This enabled visualizations and low-dimensional approximations that led to a continuous approximation of the Pareto frontier for the multi-criteria problem of maximizing all convex combinations of network throughputs over MAC and PHY parameters. Such a result simplifies and regularizes the search for parameters that enable high quality performance of both networks, particularly compared to approaches that do not operate on a reduced parameter space. Analysis of the LAA-Wi-Fi example revealed an explainable and interpretable solution to an otherwise challenging problem---\emph{devoid of any known convexity until subsequent exploration}. 

Future work will incorporate alternative low-dimensional approximations including both cases of Grassmannian mixing and subspace unions to improve the trace. We will also study alternative methods of subspace and non-linear dimension reduction to accelerate reinforcement learning over the various near-optimal Pareto manifolds. Future approaches will enable spectrum sharing for a variety of wireless communications models over unlicensed bands by simplifying the design of wireless network operation and architecture---ultimately quantifying model parameter combinations giving near-optimal KPI trade-offs.

\bibliographystyle{IEEEtran}
\bibliography{bibliography}

%----------------------------------------------------------------------------------------

\end{document}

%% file: Intro.tex
\lettrine[nindent=0em,lines=3]{A}s wireless communications evolve and proliferate into our daily lives, the demand for radio spectrum grows dramatically. To accommodate this growth, wireless device protocols are beginning to transition from a predominantly-licensed spectrum to a shared approach in which use of the unlicensed spectrum bands is increasing rapidly. The main bottleneck of this approach, however, is balancing new network paradigms with incumbent networks, such as Wi-Fi. 

Previously, unlicensed bands were dominated by Wi-Fi traffic and, occasionally, used by commercial cellular carriers for offloading data that would otherwise have been communicated via Long-Term Evolution (LTE) in the licensed spectrum. In order to address spectrum scarcity in new operating paradigms, mobile network operators are choosing to operate in unlicensed bands (such as LAA) in addition to data offloading \cite{3GPP-2015}. Even though operating LAA in unlicensed bands improves spectral-usage efficiency, it could have a significant influence on Wi-Fi operation and thereby create a number of challenges for spectrum sharing. Understanding and addressing these challenges calls for a deep dive into the operations and parameter selection of both networks in the medium access control (MAC) and physical (PHY) layers.

\subsection{Related Work}
There have been many investigations of fairness in spectrum sharing among LAA and Wi-Fi networks \cite{Quek-Access2016,Cano-Acm2017,Mehrnoush-2018}. Critically, these works do not consider optimizing key performance indicators (KPIs). In contrast to \cite{Quek-Access2016,Cano-Acm2017,Mehrnoush-2018}, the authors in \cite{Gao-2017} and \cite{Gao-2019} maximize LAA throughput and total network sum rate, respectively, over contention window sizes of both networks while guaranteeing the Wi-Fi throughput satisfies a threshold. Ignoring the constraint on Wi-Fi throughput, the authors in \cite{Mengqi-2020} maximize the overall network throughput over the same variables---contention window size of both networks. These studies \cite{Gao-2017, Gao-2019, Mengqi-2020}, however, optimize only a single MAC layer parameter and do not consider optimizing over a set of MAC and PHY layer parameters. A multi-criteria optimization problem was formulated in \cite{Yin-2016} to satisfy the quality of service requirements of LAA eNodeBs by investigating the trade-off between the co-channel interference in the licensed band and the Wi-Fi collision probability in the unlicensed band. The line of work in \cite{Yin-2016} is further expanded in
\cite{Mosleh-GC-2019}.  Considering both PHY and MAC layer parameters, \cite{Mosleh-GC-2019} maximizes the weighted sum rate of an LAA network subject to Wi-Fi throughput constraint with respect to the fraction of time that LAA is active. However, Wi-Fi throughput was not simultaneously optimized in \cite{Yin-2016} nor \cite{Mosleh-GC-2019}.

In its most general form, the spectrum sharing problem can be modeled as a multi-criteria optimization problem where the KPIs of all operators, in a heterogeneous network that coexists on the unlicensed band, are maximized simultaneously. In this context, we explore optimal trade-offs between Wi-Fi and LAA throughputs that are simultaneously maximized over an aggregate of PHY and MAC layer parameters. This motivates a multi-criteria optimization formalism in which the input space has high-dimensionality \cite{Grey2021}. The model we investigated in \cite{Grey2021}, and in this paper as well, uses $17$ MAC and PHY layer variables to characterize the LAA and Wi-Fi coexistence performance. 

Previous experience suggests that not all of these variables are equally important in determining the quality of network KPIs. To address this difficulty, we use a mathematical formalism known as \textit{active subspaces} to determine parameter combinations that change KPI values the most on average---and those that do not \cite{Constantine2015}. The sets of parameter combinations defined by the active subspaces inform KPI approximations and visualizations over a low-dimension subspace---simplifying and regularizing the multi-criteria optimization. Here, we further expand our investigation in \cite{Grey2021} and study the impact of these parameters via analytical models and simulation results. This approach explores the system behavior and provides deep insights into related spectrum sharing and communication systems---LAA and Wi-Fi coexistence is merely an example in this work. This work not only optimizes the unlicensed band spectrum sharing, but also sheds light on the importance of network parameter selection and supplements with exploration to facilitate new explanations and interpretations of results. 

We incorporate active subspace dimension reduction into the multi-criteria optimization framework to analyze, interpret, and explain the shared spectrum coexistence problem. The set of maximizing arguments quantify the inherent trade-off between LAA and Wi-Fi throughputs. The dimension reduction supplements a trade-off analysis of network throughputs by computing a Pareto trace. The Pareto trace provides a \emph{continuous} approximation of Pareto optimal (non-dominated) points in a common domain of a multi-criteria problem \cite{boyd2004convex, bolten2020tracing}---resulting in a parameter manifold consisting of near-best trade-offs between differing throughputs. Facilitated by the dimension reduction, our work provides a continuous description of this parameter manifold that quantifies high-quality performance of both networks.

\subsection{Contributions}

This work differs from those previously mentioned by facilitating \emph{new interpretations and explanations} of results. The main contributions of this paper can be
summarized as follows:
\begin{itemize}
\item For the first time, we incorporate \emph{active subspace dimension reduction into a multi-criteria optimization} problem for radio spectrum sharing. In this application, we maximize modeled Wi-Fi and LAA network throughputs as functions over a high dimensional input space of both PHY and MAC layer parameters. We also provide subsequent \emph{ridge approximations} (defined in section \ref{sec:AS}) of the KPIs in a coexistence scenario over an unlicensed band. 
\item We determine parameter combinations that are most important in changing the quality of network KPIs and those that are not. The sets of parameter combinations achieved \emph{simplifies and regularizes the multi-criteria optimization} by informing KPI approximation and \emph{visualizations over a low-dimension parameter subspace}. 
% \item \added[id=ZG]{Active subspace approximation of throughputs in a coexistence scenario.}
%\item \added[id=ZG]{Global sensitivity analysis informed by active subspaces.}
\item We calculate \emph{convex quadratic ridge approximations} of network throughputs that inform a \emph{continuous quadratic trace} describing the trade-off between near-optimal network throughput combinations. This offers a \emph{continuous description of a parameter manifold} that quantifies high quality performance of both networks.
\item Using these convex quadratic ridge approximations, we also supplement with a numerical experiment suggesting \emph{the resulting trace is more stable} (when subjected to variations in data used for fitting) by virtue of regularizing over a low-dimensional subspace. 
\item Finally, we study the impact of the most important parameter combinations via \emph{explainable and interpretable simulation results} (visualizing the Pareto tracing). Our proposed approach explores the system behavior and provides deep insights into the optimization of unlicensed band spectrum sharing. Simulation results show that the \emph{proposed scheme is a promising candidate for improving both network throughputs in a coexistence scenario and parametrizing predominantly-flat manifolds of Pareto optimal solutions.}
\end{itemize}

%\deleted[id=SM]{The multi-criteria optimization formalism we propose is complicated by the high-dimensionality of the input space the model we investigate below uses $17$ MAC and physical layers variables to characterize the coexistence performance. Nevertheless, previous experience suggests that not all of these variables are equally important in determining the quality of network KPIs. To address this difficulty, we use a mathematical tool known as active subspaces which helps determine parameter combinations which change KPI values the most on average. The sets of parameter combinations defined by the active subspaces help inform KPI approximations and visualizations over an aggregate low-dimension subspace---simplifying the multi-criteria optimization.} 

%\deleted[id=SM]{We incorporate active subspace dimension reduction into a multi-criteria optimization framework to analyze the shared spectrum coexistence problem. The dimension reduction supplements a trade-off analysis of network throughputs by computing a Pareto trace. The Pareto trace provides a continuous approximation of Pareto optimal (non-dominated) points in a common domain of a multi-criteria problem---resulting in a near-best trade-off between differing throughputs. This offers a continuous description of a parameter subset which quantifies high quality performance of both networks, facilitated by a dimension reduction.}

\subsection{Paper Organization and Notation}
The paper is organized as follows: section \ref{sec:sys_model} describes the system model and presents the problem formulation. Section \ref{sec:pareto_tracing} formalizes solutions to the problem statement and introduces the concept of a Pareto trace. Section \ref{sec:AS} introduces active subspaces and offers technical considerations for quantifying spectrum sharing as a continuous Pareto trace of near Pareto optimal parameters. Simulation results are shown and discussed in section \ref{sec:Numerics}. Finally, we conclude with an overview of the results and remarks about future work.

\textit{Notation}: Throughout the paper, standard math-font letters are used to denote scalars.
Boldface capital and boldface lower-case letters denote matrices
and vectors, respectively. All vectors, e.g., $\boldsymbol{u}, \boldsymbol{v} \in \mathbb{R}^K$, are assumed to be tall (column) vectors with $K$ entries and all expressions correspond to standard matrix vector multiplication. The transpose of matrix $\mathbf{A}$ is
denoted by $\mathbf{A}^{\transpose}$. The operator $\odot$ represents the Hadamard product and the operator $\left( \cdot \right)_k$ represents the $k$-th index of a vector---i.e., $(\boldsymbol{u} \odot \boldsymbol{v})_k = (\boldsymbol{u})_k(\boldsymbol{v})_k$. Complementary probabilities are represented entry-wise over vectors by the operator $\cdot^{\complement}$ such that $\boldsymbol{p}_{\mathcal{K}}^{\complement}= \boldsymbol{1}_{\mathcal{K}} - \boldsymbol{p}_{\mathcal{K}}$ where $\boldsymbol{1}_{\mathcal{K}}$ is an appropriately sized vector of ones. The gradient $\nabla$ and Hessian $\nabla^2$ are taken with respect to model parameters $\boldsymbol{\theta}$, if not otherwise decorated by a label. Finally, we define sets with calligraphic letters $\mathcal{K}$ and, specifically, $\mathcal{K}$ is used regularly as a placeholder. Parameters which are varied or examined for the purposes of model exploration, optimization, and transformation are denoted using the Greek alphabet. %$\otimes$ represents the tensor (outer) product, i.e., $\boldsymbol{u} \otimes \boldsymbol{v} = \boldsymbol{u}\boldsymbol{v}^{\transpose}$ for any two column vectors. 

%A circularly symmetric
% complex Gaussian random variable (r.v.) is represented by $Z = X + jY  \sim \mathcal{CN} (0, \sigma^{2})$, where $X$ and $Y$ are independent
% and identically distributed (i.i.d.) normal r.v.’s from $\mathcal{N}(0, \sigma^{2}/2 )$.

%% file: Sys_Model_Simple.tex
We consider a downlink coexistence scenario where two mobile network operators (MNOs) operate over the same shared unlicensed industrial, scientific, and medical (ISM) radio band. We are primarily focused on the operation of cellular base stations in the unlicensed bands. However, LTE base stations may have permission to utilize a licensed band as well. We assume the MNOs use time sharing to simultaneously operate in the unlicensed band and we aim to analyze competing trade-offs in throughputs of the Wi-Fi and LAA systems. A network throughput is a function of both MAC and PHY layer parameters. In this section, we introduce the parameters defining the network topology, MAC
layer protocols, the PHY layer, and briefly discuss the relation of these variables to network throughput.

\subsection{Network Topology}
We consider a coexistence scenario in which the LAA network consists of $L$ Evolved Node B (eNodeBs) indexed by the set $\mathcal{L}$, $i \in \mathcal{L} = \lbrace 1,2,\dots,L\rbrace$, while the Wi-Fi network is composed of
$W$ access points (APs) indexed by the set $\mathcal{W}$, $j \in \mathcal{W} = \lbrace 1,2,\dots,W \rbrace$. Note that our proposed subspace-based Pareto tracing approach could be applied to many types of communication systems (or alternative applications), but LTE is used here as an example. The eNodeBs and APs are randomly distributed over a rectangular area while LAA user equipment (UEs) and Wi-Fi clients/stations (STAs) are, respectively, distributed around each eNodeB and AP independently and uniformly. Each transmission node serves one single antenna UEs/STAs. %We assume any transmission
% node \deleted[id=ZG]{indexed by $\mathcal{L}$ or $\mathcal{W}$ transmits with $P_{k}$, $k \in \{\mathcal{L}, \mathcal{W}\}$, and the} and user association is based on the received power. 
We assume \textit{(i)} both Wi-Fi and LAA are in the saturated traffic condition, i.e., at least one packet is waiting to be sent, \textit{(ii)} there are neither hidden nodes nor false alarm/miss detection problems in the network\footnote{We assume perfect spectrum sensing in both systems. The impact of imperfect sensing is beyond the scope of this paper and investigating the effect of sensing errors is an important topic for future work.}, \textit{(iii)} the channel knowledge is ideal, so the only source of unsuccessful transmission is collision, \textit{(iv)} a successful transmission happens if only one link transmits at a time, i.e.,
exclusive channel access (ECA) model is considered, and \textit{(v)}  each link is subject to Rayleigh fading and Log-normal shadowing.

\subsection{MAC Layer Protocols}
The medium access key feature in both Wi-Fi and LAA involves the station accessing the medium to sense the channel by performing clear channel assessment prior to transmitting. The station only transmits if the medium is determined to be idle. Otherwise, the transmitting station refrains from transmitting data until it senses the channel is available. Although LAA and Wi-Fi technologies follow similar channel access procedures, they utilize different carrier sense schemes, different channel sensing threshold levels, and different channel contention parameters, leading to different unlicensed channel access probabilities and thus, different throughputs.

%Here, we will briefly describe the network throughput of both systems, a quantity that will be used in analyzing competing trade-offs in throughputs of the Wi-Fi and LAA systems later.
Conforming with the analytical model in \cite{Bianchi-2000}, the probability of either network transmitting a packet in a randomly chosen time slot can be expressed as 
\begin{equation} \label{Eqn_1}
p_{k} =\\ \frac{2(1 - 2 c_{k})}{(1-2c_{k})(1 + \omega_{k}) + c_{k}\omega_{k}(1-(2c_{k})^{\mu_{k}})}, 
\end{equation}
where $k$ is representative of an index from either $\mathcal{L}$ or $\mathcal{W}$, $c_{k}$ denotes the probability of collision experienced by the $k$-th transmission node, and $\omega_{k}$ and $\mu_{k}$ indicate the minimum contention window size and the maximum back-off stage, respectively, of the $k$-th transmission node on the unlicensed channel. 

\begin{table*}[t]
\centering
\captionsetup{font=footnotesize}
 \caption{Probabilities of time slot allocations}\label{tbl:probs}
 \begin{threeparttable}
 \centering
 \footnotesize
  \begin{tabular}{c|p{8cm} |l}
    \hline
    Probability & Description & Expression \\ \hline
    %\rowcolor{LightCyan}
     $(\boldsymbol{p}_{\mathcal{T}})_1$ & Probability unlicensed band is idle & %$\det\left( \text{diag}\left(\boldsymbol{p}_{\mathcal{W}}^{\complement}\right)\right) \det\left(\text{diag}\left(\boldsymbol{p}_{\mathcal{L}}^{\complement}\right)\right)$
     $z_{\mathcal{W}}z_{\mathcal{L}}$\\[0.1cm] %\hline
      %\rowcolor{LightCyan}
    $(\boldsymbol{p}_{\mathcal{T}})_2$ & Probability of successful Wi-Fi transmission on the unlicensed band & $ \boldsymbol{p}_{\mathcal{W}}^{\transpose}\boldsymbol{c}_{\mathcal{W}}^{\complement}$\\[0.1cm] %\hline
      %\rowcolor{LightCyan}
      $(\boldsymbol{p}_{\mathcal{T}})_3$ & Probability of successful LAA transmission on the unlicensed band & $ \boldsymbol{p}_{\mathcal{L}}^{\transpose}\boldsymbol{c}_{\mathcal{L}}^{\complement}$\\[0.1cm] %\hline
      %\rowcolor{LightCyan}
      $(\boldsymbol{p}_{\mathcal{T}})_4$ & Probability of collision among the Wi-Fi transmissions & %$\det\left( \text{diag}(\boldsymbol{p}_{\mathcal{L}}^{\complement})\right)\left[1- \det\left( \text{diag}(\boldsymbol{p}_{\mathcal{W}}^{\complement})\right) -  \boldsymbol{\iota}_{\mathcal{W}}^{\transpose}\boldsymbol{p}_{\mathcal{W}}\right]$
      $z_{\mathcal{L}}\left[1- z_{\mathcal{W}} -  z_{\mathcal{W}}\sum_{j\in\mathcal{W}} (\boldsymbol{p}_{\mathcal{W}})_j/(\boldsymbol{p}^{\complement}_{\mathcal{W}})_j\right]$\\[0.1cm] %\hline
      %\rowcolor{LightCyan}
      $(\boldsymbol{p}_{\mathcal{T}})_5$ & Probability of collision among the LAA transmissions & %$\det\left( \text{diag}(\boldsymbol{p}_{\mathcal{W}}^{\complement})\right)\left[1- \det\left( \text{diag}(\boldsymbol{p}_{\mathcal{L}}^{\complement})\right) - \boldsymbol{\iota}_{\mathcal{L}}^{\transpose}\boldsymbol{p}_{\mathcal{L}}\right]$
      $z_{\mathcal{W}} \left[1- z_{\mathcal{L}} -  z_{\mathcal{L}}\sum_{i\in\mathcal{L}} (\boldsymbol{p}_{\mathcal{L}})_i/(\boldsymbol{p}^{\complement}_{\mathcal{L}})_i\right]$\\[0.1cm] %\hline
      %\rowcolor{LightCyan}
      $(\boldsymbol{p}_{\mathcal{T}})_6$ & Probability of collision among Wi-Fi \& LAA transmissions &%$\left[1 - \det\left( \text{diag}(\boldsymbol{p}_{\mathcal{L}}^{\complement})\right)\right]\left[1 - \det\left( \text{diag}(\boldsymbol{p}_{\mathcal{W}}^{\complement})\right)\right]$
      $\left[ 1 - z_{\mathcal{L}}\right]\left[1 -  z_{\mathcal{W}}  \right]$%\hline
      %\rowcolor{LightCyan}
    \end{tabular}
    \begin{tablenotes}
           \item \text{Note: for brevity, we substitute $z_{\mathcal{K}} = \prod_{k \in \mathcal{K}} (\boldsymbol{p}^{\complement}_{\mathcal{K}})_k$ as the overall probability of transmission for $\mathcal{K}$ representing either}
           \item \text{network $\mathcal{L}$ or $\mathcal{W}$.}
    \end{tablenotes}
   \end{threeparttable}
\end{table*}
\normalsize

To simplify notation, we aggregate the stationary transmission probability model \eqref{Eqn_1} into entries of a vector $\boldsymbol{p}$. With this notation, the Wi-Fi stationary transmission probability of AP $j \in \mathcal{W}$ is considered $(\boldsymbol{p})_j:(c_{j}, \omega_{j}, \mu_{j})\mapsto p_{j}(c_{j}; \omega_{j}, \mu_{j})$ such that $\boldsymbol{p}_{\mathcal{W}} \overset{\Delta}{=} \boldsymbol{p}(\boldsymbol{c}_{\mathcal{W}}; \boldsymbol{\omega}_{\mathcal{W}}, \boldsymbol{\mu}_{\mathcal{W}}) \in \mathbb{R}^{W}$ represent the set of all Wi-Fi probabilities for all APs. The Wi-Fi probabilities depend explicitly on $\boldsymbol{c}_{\mathcal{W}} \in \mathbb{R}^{W}$ and parameters $\boldsymbol{\omega}_{\mathcal{W}},\boldsymbol{\mu}_{\mathcal{W}} \in \mathbb{R}^{W}$. Similarly, for the LAA eNodeB's, we assign $\boldsymbol{p}_{\mathcal{L}} \overset{\Delta}{=} \boldsymbol{p}(\boldsymbol{c}_{\mathcal{L}}; \boldsymbol{\omega}_{\mathcal{L}}, \boldsymbol{\mu}_{\mathcal{L}})$ such that $\boldsymbol{p}_{\mathcal{L}}, \boldsymbol{c}_{\mathcal{L}},\boldsymbol{\omega}_{\mathcal{L}},\boldsymbol{\mu}_{\mathcal{L}} \in \mathbb{R}^{L}$. Note, for brevity, we are dropping the explicit dependencies, $p_{k}(c_{k}; \omega_{k}, \mu_{k})$, and supplement with an index-set subscript, $\boldsymbol{p}_{\mathcal{K}}$, to indicate the length of the vector-valued map as the cardinality of $\mathcal{K}$ in addition to the network association when $\mathcal{K}$ is either $\mathcal{W}$ or $\mathcal{L}$. Moreover, these explicit dependencies are conflated by a set of complementary probabilities, $\boldsymbol{p}_{\mathcal{K}}^{\complement}= \boldsymbol{1}_{\mathcal{K}} - \boldsymbol{p}_{\mathcal{K}}$.

We write the LAA overall probability of transmission as $1 - \prod_{i\in \mathcal{L}} (\boldsymbol{p}^{\complement}_{\mathcal{L}})_i$ and similarly for Wi-Fi $1 -\prod_{j\in \mathcal{W}}(\boldsymbol{p}_{\mathcal{W}}^{\complement})_j$. In a complementary fashion, the collision probability of the transmitting Wi-Fi AP $j\in \mathcal{W}$ and LAA eNodeB $i\in \mathcal{L}$ on a shared unlicensed band are expressed as entries
\begin{align} \label{eq:Wifi_collision}
    (\boldsymbol{c}_{\mathcal{W}})_j = 1 -\frac{1}{(\boldsymbol{p}^{\complement}_{\mathcal{W}})_j}\prod_{k\in \mathcal{W}}(\boldsymbol{p}^{\complement}_{\mathcal{W}})_k \prod_{i\in \mathcal{L}}(\boldsymbol{p}^{\complement}_{\mathcal{L}})_i,
\end{align}
and
\begin{align} \label{eq:LAA_collision}
    (\boldsymbol{c}_{\mathcal{L}})_i = 1 -\frac{1}{(\boldsymbol{p}^{\complement}_{\mathcal{L}})_i}\prod_{k\in \mathcal{L}}(\boldsymbol{p}^{\complement}_{\mathcal{L}})_k \prod_{j\in \mathcal{W}}(\boldsymbol{p}^{\complement}_{\mathcal{W}})_j,
\end{align}
respectively, stored in the vectors $\boldsymbol{c}_{\mathcal{W}} \in \mathbb{R}^{W}$ and $\boldsymbol{c}_{\mathcal{L}} \in \mathbb{R}^{L}$. Notice, \eqref{eq:Wifi_collision} and \eqref{eq:LAA_collision} are expressions dependent on both $\boldsymbol{p}_{\mathcal{L}}$ and $ \boldsymbol{p}_{\mathcal{W}}$ given parameters $\boldsymbol{\omega}_{\mathcal{W}}, \boldsymbol{\mu}_{\mathcal{W}},\boldsymbol{\omega}_{\mathcal{L}}, \boldsymbol{\mu}_{\mathcal{L}}$. Consequently, equations \eqref{eq:Wifi_collision} and \eqref{eq:LAA_collision} represent a coupling of probabilities. %as a vector-valued map $\boldsymbol{c}_{\mathcal{W}} \overset{\Delta}{=} \boldsymbol{c}(\boldsymbol{p}_{\mathcal{W}}, \boldsymbol{p}_{\mathcal{L}}; \boldsymbol{\omega}_{\mathcal{W}}, \boldsymbol{\mu}_{\mathcal{W}},\boldsymbol{\omega}_{\mathcal{L}}, \boldsymbol{\mu}_{\mathcal{L}})$ where the first input  $\boldsymbol{p}_{\mathcal{W}}$ corresponds to the label, $\mathcal{W}$, and the form of $\boldsymbol{c}$ is inferred from \eqref{eq:Wifi_collision} or \eqref{eq:LAA_collision}---similarly for the LAA network, $\boldsymbol{c}_{\mathcal{L}} \overset{\Delta}{=} \boldsymbol{c}(\boldsymbol{p}_{\mathcal{L}}, \boldsymbol{p}_{\mathcal{W}}; \boldsymbol{\omega}_{\mathcal{W}}, \boldsymbol{\mu}_{\mathcal{W}},\boldsymbol{\omega}_{\mathcal{L}}, \boldsymbol{\mu}_{\mathcal{L}})$.

\subsection{Model Computational Details}
Given the coupling induced by \eqref{eq:Wifi_collision} and \eqref{eq:LAA_collision}, $\boldsymbol{p}_{\mathcal{L}}$ and $\boldsymbol{p}_{\mathcal{W}}$ now implicitly depend on all probabilities and parameters. This leads to a simultaneous non-linear system of four vector-valued equations, $\boldsymbol{p}_{\mathcal{W}}, \boldsymbol{c}_{\mathcal{W}} \in [0,1]^{W}$ and $\boldsymbol{p}_{\mathcal{L}}, \boldsymbol{c}_{\mathcal{L}} \in [0,1]^{L}$. This system consists of $2(W + L)$ equations and unknowns: $\boldsymbol{p}_{\mathcal{L}}, \boldsymbol{p}_{\mathcal{W}}, \boldsymbol{c}_{\mathcal{L}}, \boldsymbol{c}_{\mathcal{W}}$ given parameters: $\boldsymbol{\omega}_{\mathcal{L}}, \boldsymbol{\omega}_{\mathcal{W}}, \boldsymbol{\mu}_{\mathcal{L}}, \boldsymbol{\mu}_{\mathcal{W}}$. We solve the non-linear system using a trust-region method \cite{powell1968fortran, more1980user} according to provided parameters $\boldsymbol{\omega}_{\mathcal{L}}, \boldsymbol{\mu}_{\mathcal{L}} \in \mathbb{R}^{L}$ and $\boldsymbol{\omega}_{\mathcal{W}}, \boldsymbol{\mu}_{\mathcal{W}} \in \mathbb{R}^{W}$ (as model inputs). The result is a set of consistent Wi-Fi AP and LAA eNodeB probabilities of transmission and collision---with numerical implementation constituting a map from parameters to probabilities (MAC layer model).

The probability of a successful transmission for the Wi-Fi AP $j$ (resp. LAA eNodeB $i$) on the unlicensed band is the $j$-th entry of $\boldsymbol{p}_{\mathcal{W}}\odot \boldsymbol{c}_{\mathcal{W}}^{\complement}$ (resp. $i$-th entry of $\boldsymbol{p}_{\mathcal{L}}\odot\boldsymbol{c}_{\mathcal{L}}^{\complement}$). Additionally, the average duration to support one successful transmission in the
unlicensed band can be calculated as
$
\boldsymbol{p}_{\mathcal{T}}^{\transpose}\boldsymbol{t}
$
where $\boldsymbol{p}_{\mathcal{T}}$ is defined in Table \ref{tbl:probs}. The vector $\boldsymbol{t}$ is dictated by an access mechanism. We consider the basic access mechanism as in \cite{Bianchi-2000,Yao-2019} and, therefore, entries of $\boldsymbol{t}$ are defined as $\boldsymbol{t} = (T_{idle}, T_{s,\mathcal{W}}, T_{s,\mathcal{L}}, T_{c, \mathcal{W}}, T_{c,\mathcal{L}}, T_{c,\mathcal{W}\mathcal{L}})^{\transpose}$ per the notation and corresponding computations in \cite{Mosleh-VTC2020}.

\subsection{Physical Layer Parameters and Data Rates}
To calculate network throughput, we also need to introduce data rates parametrized by physical layer parameters. The achievable physical data rate of the $\mathcal{L}$ or $\mathcal{W}$ operators is a function of link signal-to-noise ratio (SNR) that depends on changes with the link distances and propagation model \cite{Mosleh-VTC2020}. The link distances and propagation model admit parameter dependencies that are quantified in the subset of PHY layer parameters in $\boldsymbol{\theta}$. The subsequent data rate dependencies are expressed as $\log_{2}(1+\text{SNR}_{\mathcal{K}}(\boldsymbol{\theta}; \boldsymbol{x}))$
%\begin{align}
 %   \delta r_{\mathcal{K}}(\boldsymbol{\theta}; \boldsymbol{x}) &= 
%\end{align}
where $\boldsymbol{\theta}$ and $ \boldsymbol{x}$ are vectors representing the full set of parameters that can be varied to study the model behavior and the remaining \emph{fixed} scenario parameters, respectively. These parameter sets are discussed in detail in the following sub-section.

\subsection{Parametrized Model}
We parametrize MAC layer parameters by assuming common minimum contention window sizes and maximum back-off stages for each network. In other words, $\boldsymbol{\omega}_{\mathcal{W}} = \theta_1\boldsymbol{1}_{\mathcal{W}}$, $\boldsymbol{\omega}_{\mathcal{L}} = \theta_2\boldsymbol{1}_{\mathcal{L}}$, $\boldsymbol{\mu}_{\mathcal{W}} = \theta_3\boldsymbol{1}_{\mathcal{W}}$, and $\boldsymbol{\mu}_{\mathcal{L}} = \theta_4\boldsymbol{1}_{\mathcal{L}}$ for parameters $\theta_k \in \mathbb{R}$. It is conceivable that we may consider each independently for a total of $2(L + W)$ parameters as entries in a vector $\boldsymbol{\theta}$. However, we opt for a simplification to four total parameters---two common parameters per network. These parametrizations result in subsequent dependencies. Thus, in general, we consider $\boldsymbol{\mu}_{\mathcal{K}}(\boldsymbol{\theta})$ and $\boldsymbol{\omega}_{\mathcal{K}}(\boldsymbol{\theta})$ where $\mathcal{K}$ is either $\mathcal{W}$ or $\mathcal{L}$ and $\boldsymbol{\theta}\in \mathbb{R}^{D}$ is a vector of all parameters---the first four defined as MAC parameters and the remaining parameters described below and in Table \ref{Tb1}. Notice that certain partial derivatives are zero in this general vector-valued map interpretation, e.g., $\partial (\boldsymbol{\omega}_{\mathcal{W}})_j /\partial (\boldsymbol{\theta})_k =0$ for all $k=2,\dots,D$ and for all $j=1,\dots,W$. This may suggest degeneracy in computations of the Jacobians that subsequently inform gradients of KPIs via the chain rule. Hence, certain directions in the parameter space may be more or less informative in changing various KPI predicted by this model. Regardless, we can now reinterpret the various probabilities as dependent maps $\boldsymbol{p}_{\mathcal{K}}(\boldsymbol{\theta})$ and $\boldsymbol{c}_{\mathcal{K}}(\boldsymbol{\theta})$ for either network $\mathcal{K}$ that depend on variations in $\boldsymbol{\theta}$. Notice any change in the input parameters $\boldsymbol{\theta}$ requires a new solution to the non-linear system. Solving the non-linear system constitutes the majority of the computational burden in the model given new parameters $\boldsymbol{\theta}$ as input.

We also append the physical parameters to the MAC parameter vector by reassigning $\boldsymbol{\theta}$ as a vector representing the full set of parameters that can be varied to study the model behavior. These additional PHY parameters are summarized with appropriate bounds and description as the remaining entries of $\boldsymbol{\theta}$ in Table \ref{Tb1}. Any remaining physical parameters are held fixed and constitute a \textit{scenario} for a particular model evaluation. We aggregate these remaining scenario parameters into a vector $\boldsymbol{x}$, summarized in Table \ref{Tb1}.

The LAA and Wi-Fi throughputs, indicated respectively by $f_{\mathcal{L}}$ and $f_{\mathcal{W}}$, are functions of $D$ total MAC and PHY layer parameters in a vector $\boldsymbol{\theta}$ conditioned on fixed values in a vector $\boldsymbol{x}$,
\begin{align} \label{eq:throughputs}
f_{\mathcal{L}}&: \mathbb{R}^D \times \lbrace \boldsymbol{x} \rbrace \rightarrow \mathbb{R}:(\boldsymbol{\theta}, \boldsymbol{x}) \mapsto f_{\mathcal{L}}(\boldsymbol{\theta}; \boldsymbol{x}),\\ \nonumber
f_{\mathcal{W}}&: \mathbb{R}^D\times \lbrace \boldsymbol{x} \rbrace \rightarrow \mathbb{R}:(\boldsymbol{\theta}, \boldsymbol{x}) \mapsto f_{\mathcal{W}}(\boldsymbol{\theta}; \boldsymbol{x}).
\end{align}
Given a set of consistent probabilities for parameters $\boldsymbol{\theta}$, the throughputs are computed as
\begin{equation}\label{eq:throughput_full}
     f_{\mathcal{K}}(\boldsymbol{\theta}; \boldsymbol{x}) = T_{\mathcal{K}}\log_{2}(1+\text{SNR}_{\mathcal{K}}(\boldsymbol{\theta}; \boldsymbol{x}))\frac{\boldsymbol{p}_{\mathcal{K}}^{\transpose}(\boldsymbol{\theta})\boldsymbol{c}^{\complement}_{\mathcal{K}}(\boldsymbol{\theta})}{\boldsymbol{t}^{\transpose}\boldsymbol{p}_{\mathcal{T}}(\boldsymbol{\theta})},
\end{equation}
for arbitrary $\mathcal{K}$ representing either $\mathcal{L}$ or $\mathcal{W}$, where  $T_{\mathcal{K}}$ indicates the (Wi-Fi or LAA) operator payload duration. We note that the numerator in \eqref{eq:throughput_full} can be rewritten as the total probability of successful transmission---similarly for the second and third entries of $\boldsymbol{p}_{\mathcal{T}}(\boldsymbol{\theta})$. This facilitates an intuitive explanation that throughput is simply the proportion of successful transmissions to all remaining transmission events if data rates defining $\text{SNR}_{\mathcal{K}}$ are held constant.

For this application, LAA throughput $f_{\mathcal{L}}$ and Wi-Fi throughput $f_{\mathcal{W}}$ are only considered functions of $D$ variable parameters in $\boldsymbol{\theta}$. This numerical study and choice of model considers $D=17$. However, parameters and models will be further generalized as part of on-going research efforts.

\subsection{Problem Formulation}
The problem of interest is to maximize a convex combination of network throughputs for the fixed scenario $\boldsymbol{x}$ over the MAC and PHY parameters $\boldsymbol{\theta}$ in a multi-criteria optimization. A \emph{Pareto front} is quantified by the following optimization problem:
\begin{equation} \label{eq:multiopt}
    \underset{\boldsymbol{\theta} \in \mathcal{D}  \subset \mathbb{R}^D}{\text{maximize}} \,\,tf_{\mathcal{L}}(\boldsymbol{\theta}; \boldsymbol{x}) + (1-t)f_{\mathcal{W}}(\boldsymbol{\theta}; \boldsymbol{x}),
\end{equation}
for all $t \in [0,1]$ where $\mathcal{D}$ is the parameter domain defined by the ranges in Table \ref{Tb1}. We refer to $t=0$ and $t=1$ as \emph{left} and \emph{right} solutions, respectively. The goal is to quantify a smooth trajectory $\boldsymbol{\theta}(t;  \boldsymbol{x})$ through MAC and PHY parameter space, or \emph{trace} \cite{bolten2020tracing}, such that the convex combination of throughputs is maximized over a map $\boldsymbol{\theta}:[0,1] \times \{\boldsymbol{x}\} \rightarrow \mathcal{D}$. In Section \ref{sec:pareto_tracing} we formalize this notion of a trace. In Section \ref{sec:AS}, we summarize an exploratory approach for understanding to what extent problem \eqref{eq:multiopt} is convex \cite{boyd2004convex} and how we can intuitively regularize given the possibility of degeneracy induced by the chosen parametrization or otherwise. The empirical evidence generated through visualization and dimension reduction provide justification for convex quadratic approximations and subsequent quadratic trace in Section \ref{sec:Numerics}.

\begin{table*}[t]
\centering
\captionsetup{font=footnotesize}
 \caption{MAC and PHY parameters influencing throughputs }\label{Tb1}
 \begin{threeparttable}
 \centering
 \scalebox{0.98}{
 \scriptsize
  \begin{tabular}{c|p{3.4cm} |c|c}
    \hline
    Param. & Description &Bounds & Nominal \\ \hline
    \rowcolor{paleaqua}
     $(\boldsymbol{\theta})_1$ & Wi-Fi min contention window size&(8, 1024)& 516\\[0.1cm] %\hline
      \rowcolor{paleaqua}
  $(\boldsymbol{\theta})_2$ &  LAA min contention window size&(8, 1024)& 516\\[0.1cm] %\hline
   \rowcolor{paleaqua}
    $(\boldsymbol{\theta})_3$ & Wi-Fi max back-off stage &(0, 8)& 4\\[0.1cm] %\hline
     \rowcolor{paleaqua}
    $(\boldsymbol{\theta})_4$ & LAA max back-off stage&(0, 8)& 4\\[0.1cm] %\hline
    %\rowcolor{LightCyan}
    $(\boldsymbol{\theta})_{5}$ & Distance between transmitters &(10 m, 20 m)& 15 m\\[0.1cm]
    %\rowcolor{LightCyan}
    $(\boldsymbol{\theta})_{6}$ & Minimum distance between transmitters and receivers&(10m, 35 m)& 22.5 m\\[0.1cm]
    %\rowcolor{LightCyan}
    $(\boldsymbol{\theta})_{7}$ & Height of each LAA eNodeB and Wi-Fi AP&(3 m, 6 m) & 4.5 m\\[0.1cm]
    %\rowcolor{LightCyan}
    $(\boldsymbol{\theta})_{8}$ & Height of each LAA UEs and Wi-Fi STAs & (1 m, 1.5 m) & 1.25 m\\[0.1cm]
        %\rowcolor{LightCyan}
    $(\boldsymbol{\theta})_9$ & Standard deviation of shadow fading& (8.03, 8.29) & 8.16 \\[0.1cm]
    %\rowcolor{LightCyan}
    $(\boldsymbol{\theta})_{10}$ & $k_{\text{LOS}}^{\star}$
    %LoS path-loss constant$^\star$
    & (45.12, 46.38)& 45.75\\[0.1cm]%$32.4+20\log_{10}(5)$& $83.5+20\log_{10}(5)$
    %\rowcolor{LightCyan}
    $(\boldsymbol{\theta})_{11}$ & $k_{\text{NLOS}}^{\star}$
    %Non-LoS path-loss constant$^\star$
    & (34.70, 46.38) & 40.54\\[0.1cm]%$17.3+24.9\log_{10}(5)$& $141.4+20\log_{10}(5)$
   %$x_{12}$ & standard deviation of  &--&--\\[0.1cm]
  \end{tabular} \hspace{0.02in}
  \begin{tabular}{c|p{3.4cm} |c|c}
    \hline
    Param. & Description &Bounds & Nominal \\ \hline
    %\rowcolor{LightCyan}
    $(\boldsymbol{\theta})_{12}$ &  $\alpha_{\text{LoS}}^{\star}$ 
    %LoS path-loss exponent$^\star$
    & (17.3, 21.5) & 19.4 \\[0.1cm]%The LoS path-loss exponent at each network
    %\rowcolor{LightCyan}
    $(\boldsymbol{\theta})_{13}$ &  $\alpha_{\text{NLoS}}^{\star}$
    %Non-LoS path-loss exponent$^\star$ 
    &(31.9, 38.3) & 35.1\\[0.1cm] %The NLoS path-loss exponent at each network 
    %\rowcolor{LightCyan}
    $(\boldsymbol{\theta})_{14}$ & Antenna gain at each transmitter &(0 dBi, 5 dBi) & 2.5 dBi\\[0.1cm]
    %\rowcolor{LightCyan}
    $(\boldsymbol{\theta})_{15}$ & Noise figure at each receiver&(5 dB, 9 dB) & 7 dB\\[0.1cm]
        %\rowcolor{LightCyan}
    $(\boldsymbol{\theta})_{16}$ & Transmit power at each LAA eNodeB and Wi-Fi AP&(18 dBm, 23 dBm) & 20.5 dBm\\[0.1cm]
    %\rowcolor{LightCyan}
    $(\boldsymbol{\theta})_{17}$ & Carrier channel bandwidth &(10 MHz, 20 MHz) & 15 MHz\\ [0.1cm]
    \rowcolor{Gray}
    $(\boldsymbol{x})_1$ & Number of LAA eNodeBs ($L$) &--& 6\\[0.1cm] %\hline
    \rowcolor{Gray}
    $(\boldsymbol{x})_2$ & Number of Wi-Fi APs ($W$) &--& 6\\[0.1cm] %\hline
    \rowcolor{Gray}
    $(\boldsymbol{x})_3$ & Number of LAA UEs &--& 6\\[0.1cm] %\hline
    \rowcolor{Gray}
    $(\boldsymbol{x})_4$ & Number of Wi-Fi STAs &--&6\\[0.1cm] %\hline
    \rowcolor{Gray}
    $(\boldsymbol{x})_5$ & Number of unlicensed channels &--& 1\\[0.1cm] %\hline
    %$(\boldsymbol{x})_6$ & Ideal slot time for LAA &--&--\\[0.1cm]
    \rowcolor{Gray}
    $(\boldsymbol{x})_{8}$ &  Scenario width &--&120 m\\[0.1cm]
    \rowcolor{Gray}
    $(\boldsymbol{x})_{9}$ &  Scenario height &--&80 m
    \end{tabular}}
     \begin{tablenotes}
     \footnotesize
           \item \text{Note: parameter ranges are established by 3GPP TS 36.213 V15.6.0 and 3GPP TR. 36.889 v13.0.0. MAC parameters are highlighted in blue.}
           \item \text{Scenario parameters are highlighted in gray. PHY parameters are not highlighted.}
      \item  \text{${}^{\star}$The path-loss for both line-of-sight (LoS) and non-LoS scenarios can be computed as $k + \alpha\log_{10}(d)$ in dB, where $d$ is the distance in}
      \item \text{meters between the transmitter and the receiver.}
    \end{tablenotes}
   \end{threeparttable}
\end{table*}
\normalsize

%% file: Pareto_tracing.tex
We refer to \eqref{eq:multiopt} as a maximization of the convex total objective or \emph{scalarization} \cite{bolten2020tracing}. In this case, we have a single degree of freedom to manipulate the scalarization parametrized by $t \in [0,1]$ such that
$
\varphi_t(\boldsymbol{\theta}; \boldsymbol{x}) = (1- t) f_{\mathcal{W}}(\boldsymbol{\theta}; \boldsymbol{x}) + t f_{\mathcal{L}}(\boldsymbol{\theta}; \boldsymbol{x})
$.
In order to satisfy the necessary conditions for a (locally) Pareto optimal solution, we must determine $\boldsymbol{\theta} \in \mathbb{R}^D$ critical for $\varphi_t$. The condition, also known as \emph{stationarity} condition, requires $\nabla \varphi_t(\boldsymbol{\theta}; \boldsymbol{x}) = (1- t)\nabla f_{\mathcal{W}}(\boldsymbol{\theta}; \boldsymbol{x}) + t\nabla f_{\mathcal{L}}(\boldsymbol{\theta}; \boldsymbol{x}) = \boldsymbol{0}$. Moreover, $\boldsymbol{\theta}$ is a locally (unique) Pareto optimal solution if $\nabla^2 \varphi_t(\boldsymbol{\theta}; \boldsymbol{x}) = (1- t) \nabla^2 f_{\mathcal{W}}(\boldsymbol{\theta}; \boldsymbol{x}) + t \nabla^2 f_{\mathcal{L}}(\boldsymbol{\theta}; \boldsymbol{x})$ is (symmetric) negative definite \cite{bolten2020tracing, boyd2004convex}.

In order to find a continuous (in $t$) solution to \eqref{eq:multiopt} we first examine the aforementioned necessary conditions. Provided the set of all Pareto optimal solutions to \eqref{eq:multiopt} is convex, we can continuously parametrize this set as $t \mapsto \boldsymbol{\theta}(t)$ for all $t \in [0,1]$. As an analogy, consider $t$ as pseudo-time describing an analogous trajectory through parameter space, i.e., launching from one maximizing argument to another. This interpretation is facilitated by the following Proposition (dropping scenario parameters for brevity):

\begin{prop}\label{prop:ODE}
	Given full rank $\nabla^2   \varphi_t(\boldsymbol{\theta}(t)) \in \mathbb{R}^{D \times D}$, the one-dimensional immersed submanifold parametrized by $\boldsymbol{\theta}(t) \in \mathbb{R}^D$ for all $t \in [0,1]$ is necessarily Pareto optimal such that
	$$
	\nabla^2 \varphi_t(\boldsymbol{\theta}(t))\dot{\boldsymbol{\theta}}(t) = \nabla f_{\mathcal{W}}(\boldsymbol{\theta}(t)) - \nabla  f_{\mathcal{L}}(\boldsymbol{\theta}(t)).
	$$
\end{prop}

\begin{proof}
	Differentiating the stationarity condition by composing in pseudo-time, $\nabla \varphi_t \circ \boldsymbol{\theta}(t) = \boldsymbol{0}$, results in
\begin{align*}
\boldsymbol{0}&{}=\frac{d}{dt}\left(\nabla \varphi_t \circ \boldsymbol{\theta}(t)\right)\\
&{}= \frac{d}{dt} (1- t)\left(\nabla f_{\mathcal{W}} \circ \boldsymbol{\theta}(t) \right) + \frac{d}{dt}t\left(\nabla f_{\mathcal{L}}\circ \boldsymbol{\theta}(t)\right)\\
&{}=-\nabla f_{\mathcal{W}} \circ \boldsymbol{\theta}(t) + (1-t)\left(\nabla^2_{\boldsymbol{\theta}} f_{\mathcal{W}} \circ \boldsymbol{\theta}(t)\right)\dot{\boldsymbol{\theta}}(t)\\
&\hspace{1cm} +\nabla f_{\mathcal{L}}\circ \boldsymbol{\theta}(t) + t\left(\nabla^2_{\boldsymbol{\theta}}f_{\mathcal{L}} \circ \boldsymbol{\theta}(t)\right)\dot{\boldsymbol{\theta}}(t)\\
&{}=\nabla^2 \varphi_t(\boldsymbol{\theta}(t))\dot{\boldsymbol{\theta}}(t) - \left(\nabla f_{\mathcal{W}} \circ \boldsymbol{\theta}(t) - \nabla f_{\mathcal{L}} \circ \boldsymbol{\theta}(t)\right).
\end{align*}
\normalsize

Then, the flowout along necessary Pareto optimal solutions constitutes an immersed submanifold of $\mathbb{R}^D$ nowhere tangent to the integral curve generated by the system of differential equations (See \cite{Lee2003}, Thm. 9.20)---i.e., assuming $\nabla^2 \varphi_t(\boldsymbol{\theta}(t))$ is full rank, $\nabla^2   \varphi_t(\boldsymbol{\theta}(t))^{-1}\left(\nabla   f_{\mathcal{W}}(\boldsymbol{\theta}(t)) - \nabla   f_{\mathcal{L}}(\boldsymbol{\theta}(t))\right)$ is the infinitesimal generator of a submanifold in $\mathbb{R}^D$ of locally Pareto optimal solutions contained in the flowout.
\end{proof}

We note that the system of equations in Prop. \ref{prop:ODE}, proposed in \cite{bolten2020tracing}, constitutes a set of \textit{necessary} conditions for optimality. The utility of Prop. \ref{prop:ODE} offers an interpretation that the solution set (if it exists) constitutes elements of a submanifold in $\mathbb{R}^D$. This formalism establishes a theoretical foundation for the use of manifold learning or splines over sets of points that are approximately Pareto optimal. Our future research efforts are motivated by drawing comparisons and complementary analysis with the aforementioned methods---Prop. \ref{prop:ODE} serves as the theoretical motivation and interpretation for such efforts.

Suppose $f_{\mathcal{W}}$ and $f_{\mathcal{L}}$ are well approximated by convex quadratics as \textit{surrogates}, 
\begin{align*}
-f_{\mathcal{L}}(\boldsymbol{\theta}; \boldsymbol{x})&\approx \boldsymbol{\theta}^{\transpose} \boldsymbol{Q}_{\mathcal{L}} \boldsymbol{\theta} + \boldsymbol{a}_{\mathcal{L}}^{\transpose}\boldsymbol{\theta} + c_{\mathcal{L}} \\\nonumber
-f_{\mathcal{W}}(\boldsymbol{\theta}; \boldsymbol{x})&\approx \boldsymbol{\theta}^{\transpose} \boldsymbol{Q}_{\mathcal{W}} \boldsymbol{\theta} + \boldsymbol{a}_{\mathcal{W}}^{\transpose}\boldsymbol{\theta} + c_{\mathcal{W}},
\end{align*}
such that $\boldsymbol{Q}_{\mathcal{L}}, \, \boldsymbol{Q}_{\mathcal{W}} \in S_{++}^D$,  $\boldsymbol{a}_{\mathcal{L}},\, \boldsymbol{a}_{\mathcal{W}} \in \mathbb{R}^D$, and $c_{\mathcal{L}},\,c_{\mathcal{W}} \in \mathbb{R}$ where $S_{++}^D$ denotes the collection of $D$-by-$D$ postive definite matrices---note the change in sign convention. In this case, the Pareto trace defined in Prop. \ref{prop:ODE} can be solved in closed form,
\begin{equation} \label{eq:quad_trace}
    \boldsymbol{\theta}(t; \boldsymbol{x}) = \frac{1}{2}\left[t \boldsymbol{Q}_{\mathcal{L}} + (1-t)\boldsymbol{Q}_{\mathcal{W}}\right]^{-1}\left[(t-1) \boldsymbol{a}_{\mathcal{W}} - t\boldsymbol{a}_{\mathcal{L}}\right],
\end{equation}
referred to in this context as a \textit{quadratic trace}. The challenge in our context given the possibility of degenerate Jacobians (and subsequent possibility of degeneracy in quadratic Hessians) associated with $f_{\mathcal{W}}$ and $f_{\mathcal{L}}$ is: \emph{how can we assess conditions informing \eqref{eq:quad_trace} or regularize the solve to guarantee these conditions?} We offer an approach that assess these conditions by regularization through \textit{subspace-based dimension reduction} to inform convex quadratic surrogates as approximations satisfying Prop. \ref{prop:ODE}.

Naively (for arbitrary dimension $D$ and number of samples $N$), we can also pose a convex optimization problem over the cone of positive semi-definite matrices,
\begin{equation} \label{eq:global_spd_quad}
    \underset{\begin{matrix}\boldsymbol{Q} \in S_{+}^D\\ \boldsymbol{a} \in \mathbb{R}^D\\ c \in \mathbb{R}\end{matrix}}{\text{minimize}} \,\,\frac{1}{2}\sum_{n=1}^N\left( c + \boldsymbol{\theta}_n^{\transpose}\boldsymbol{a} + \boldsymbol{\theta}_n^{\transpose} \boldsymbol{Q} \boldsymbol{\theta}_n - f_n\right)^2.
\end{equation}
Approximations resulting from problem \eqref{eq:global_spd_quad} allow us to compute the quadratic trace \eqref{eq:quad_trace} as a global data-driven surrogate given $\{(\boldsymbol{\theta}_n, f_n)\}$ of paired parameters and function responses where $f_n$ represents the throughput response---modeled or measured---for either network; e.g., $f_n = f_{\mathcal{W}}(\boldsymbol{\theta}_n; \boldsymbol{x})$ or $f_n = f_{\mathcal{L}}(\boldsymbol{\theta}_n; \boldsymbol{x})$. Solutions to problem \eqref{eq:global_spd_quad} are described as a \emph{fit} $\mathcal{CVX}\left( \{(\boldsymbol{\theta}_n, f_n)\}\right)$. Following an approximation to \eqref{eq:global_spd_quad}, the global quadratic Pareto trace is predicated on the rank of the minimizing arguments $\boldsymbol{Q}$---in particular, the rank of the convex combination of matrices given by independent solutions to \eqref{eq:global_spd_quad} using data for either throughput, i.e., $\{\boldsymbol{Q}_{\mathcal{L}}, \boldsymbol{a}_{\mathcal{L}}, c_{\mathcal{L}}\} = \mathcal{CVX}\left( \{(\boldsymbol{\theta}_n, f_{\mathcal{L}}(\boldsymbol{\theta}_n; \boldsymbol{x}))\}\right)$ and $\{\boldsymbol{Q}_{\mathcal{W}}, \boldsymbol{a}_{\mathcal{W}}, c_{\mathcal{W}}\} = \mathcal{CVX}\left( \{(\boldsymbol{\theta}_n, f_{\mathcal{W}}(\boldsymbol{\theta}_n; \boldsymbol{x}))\}\right)$. Note that in \eqref{eq:global_spd_quad} we relax the condition of strict positive definiteness to account for potentially degenerate quadratics---which can act as alternative models suggesting low-dimensional approximations \cite{Grey2017,Li1992, Constantine2015}---ensuring the solution space of \eqref{eq:global_spd_quad} is closed \cite{boyd2004convex}. Efficient solutions (fits) to problem \eqref{eq:global_spd_quad} can be computed with ``CVX,'' an open- source package for defining and solving convex programs \cite{cvx,gb08}. 

Evidently, from equation \eqref{eq:quad_trace} we see that the stability of the quadratic trace depends on the condition number of $t\boldsymbol{Q}_{\mathcal{L}} + (1-t)\boldsymbol{Q}_W$. In other words, subsequent fits $\mathcal{CVX}\left( \{(\boldsymbol{\theta}_n, f_{\mathcal{K}}(\boldsymbol{\theta}_n; \boldsymbol{x}))\}\right)$ dependent on various data sets as input may result in small perturbations to the coefficient matrices $\boldsymbol{Q}_{\mathcal{K}}$ for either throughput.\footnote{Notice that the fits $\mathcal{CVX}\left( \{(\boldsymbol{U}_r^{\transpose}\boldsymbol{\theta}_n, f_{\mathcal{K}}(\boldsymbol{\theta}_n; \boldsymbol{x}))\}\right)$ depend on the fixed scenario $\boldsymbol{x}$ and consequently any subsequent reference to these approximations should reflect this dependency.} Consequently, the otherwise unknown (and potentially variable) conditioning of $\boldsymbol{Q}_{\mathcal{K}}$ may result in instabilities of the quadratic trace subject to small perturbations---despite a reliable constrained optimization \eqref{eq:global_spd_quad}. We desire solutions that are \emph{stable} (better conditioned) against variations in the quadratic surrogate when presented with new or perturbed data. In this work, we study numerical experiments in section \ref{sec:Numerics} emphasizing a significant improvement after dimension reduction. Formalizing improvements in stability induced by dimension reduction is a topic for future research.

%% file: AS_intro.tex
Following the development in \cite{Constantine2015}, we introduce a dimension-reduction method to approximate the functions in \eqref{eq:multiopt}. This dimension-reduction technique draws from active subspace analysis to identify linear subspaces of parameters that lead to the most significant change in a function. To describe the details of the active subspaces approach, we introduce a scalar-valued function $f_{\mathcal{K}}:\mathcal{D} \subset \mathbb{R}^D \rightarrow \mathbb{R}$ defined on a compact domain $\mathcal{D}$. Again, where $\mathcal{K}$ denotes either throughput functions $f_{\mathcal{L}}$ or $f_{\mathcal{W}}$ and dropping scenario parameters for brevity.

The main results of the section rely on an eigendecomposition of the symmetric positive semi-definite matrix $\boldsymbol{C}_{\mathcal{K}}\in\mathbb{R}^{D\times D}$ defined as
\begin{equation} \label{eq:C}
\boldsymbol{C}_{\mathcal{K}} = \int_{\mathcal{D}} \nabla f_{\mathcal{K}}(\boldsymbol{\theta}) \nabla f_{\mathcal{K}}^{\transpose}(\boldsymbol{\theta}) d\boldsymbol{\theta}.
\end{equation}
In the throughput maximization application discussed in this article, the compact domain  $\mathcal{D}$ is a $D$-dimensional rectangle constructed from the Cartesian product of lower and upper bounds, $(\boldsymbol{\theta}_{\ell})_i \leq \theta_i \leq (\boldsymbol{\theta}_{u})_i$ for all $i=1,\dots,D$. We also assume a uniform measure for integration over $\mathcal{D}$. Additional discussion is available in \cite{Grey2021}.

If $\text{rank}(\boldsymbol{C}_{\mathcal{K}}) = r < D$, its eigendecomposition $\boldsymbol{C}_{\mathcal{K}} = \boldsymbol{W}_{\mathcal{K}}\boldsymbol{\Lambda}_{\mathcal{K}} \boldsymbol{W}_{\mathcal{K}}^{\transpose}$ with orthogonal $\boldsymbol{W}_{\mathcal{K}}$ satisfies $\boldsymbol{\Lambda}_{\mathcal{K}} = \mathrm{diag}(\lambda_1,\dots,\lambda_D)$ with 
\begin{equation} \label{eq:eig_decay}
    \lambda_1 \geq \lambda_2 \geq \dots \geq \lambda_r > \lambda_{r+1} = \dots = \lambda_D = 0.
\end{equation}
This defines two sets of important
$
\boldsymbol{W}_{r,\mathcal{K}} = [\boldsymbol{w}_1 \dots \boldsymbol{w}_r]_{\mathcal{K}} \in \mathbb{R}^{D\times r}
$
and unimportant
$
\boldsymbol{W}_{r,\mathcal{K}}^{\perp} = [\boldsymbol{w}_{r+1} \dots \boldsymbol{w}_D]_{\mathcal{K}} \in \mathbb{R}^{D\times (D-r)}
$
\emph{directions} over the domain. The column span of $\boldsymbol{W}_{r,\mathcal{K}}$ and $\boldsymbol{W}_{r,\mathcal{K}}^{\perp}$ constitute the \emph{active} and \emph{inactive} subspaces, respectively. Note that \eqref{eq:C} depends on a single scalar-valued response and, therefore, the active subspaces potentially differ for the separate throughputs in \eqref{eq:throughputs}.

What do we mean by \emph{important directions}? Using the eigenvectors $\boldsymbol{w}_j \in \mathbb{R}^D$, we can simplify $\boldsymbol{w}_j^{\transpose}\boldsymbol{C}_{\mathcal{K}}\boldsymbol{w}_j$ to obtain an expression for the eigenvalues,
\begin{equation}\label{eq:eigs}
    \lambda_j = \int_{\mathcal{D}}\left(\boldsymbol{w}_j^{\transpose} \nabla f_{\mathcal{K}}(\boldsymbol{\theta})\right)^2 d\boldsymbol{\theta}.
\end{equation}
In other words, the $j$-th eigenvalue can be interpreted as the mean squared directional derivative of $f_{\mathcal{K}}$ in the direction of $\boldsymbol{w}_j \in \mathbb{R}^D$ \cite{Constantine2015,Grey2021}. Precisely,
the directional derivative can be written $df_{\mathcal{K}}(\boldsymbol{\theta})[\boldsymbol{w}] = \boldsymbol{w}^{\transpose} \nabla f_{\mathcal{K}}(\boldsymbol{\theta})$, similarly $df^2_{\mathcal{K}}(\boldsymbol{\theta})[\boldsymbol{w}] = (\boldsymbol{w}^{\transpose} \nabla f_{\mathcal{K}}(\boldsymbol{\theta}))^2$,
and we obtain $\lambda_j = \mathbb{E}[df^2_{\mathcal{K}}[\boldsymbol{w}_j]]$ for expectation defined over parameters $\boldsymbol{\theta}$.
Thus, the ordering \eqref{eq:eig_decay} of the eigenpairs $\lbrace(\lambda_j, \boldsymbol{w}_j)\rbrace_{j=1}^D$ indicate directions $\boldsymbol{w}_j$ over which the function $f_{\mathcal{K}}$ changes more, on average, up to the $r+1,\dots,D$ directions that \emph{do not change the function at all} \cite{Constantine2015}. In fact, either throughput response from \eqref{eq:throughputs} is referred to as a \emph{ridge function} over $\boldsymbol{\theta}$'s if and only if $df_{\mathcal{K}}(\boldsymbol{\theta})[\boldsymbol{w}] = 0$ for all $\boldsymbol{w} \in \text{Null}(\boldsymbol{W}_{r,\mathcal{K}}^{\transpose})$.  We formalize this interpretation in the following:

\begin{prop}\cite{Constantine2015} \label{prop:ridge}
Given $\boldsymbol{w}_j$ for $j=r+1,\dots,D$ as the trailing eigenvectors of $\boldsymbol{C}_{\mathcal{K}}$, the paired eigenvalues $\lambda_j = 0$ if and only if $f_{\mathcal{K}}$ does not change over $\mathrm{span}\lbrace \boldsymbol{w}_{r+1},\dots,\boldsymbol{w}_D\rbrace$.
\end{prop}
\begin{proof}
An alternative presentation of the result is offered for completeness. ($\implies$) Take a linear combination of any two eigenvectors $\boldsymbol{w}_i$ and $\boldsymbol{w}_j$ with corresponding identically zero eigenvalues, $\lambda_i$ and $\lambda_j$, for any $i,j \in \lbrace r+1, \dots, D\rbrace$. Then, for arbitrary $a,b \in \mathbb{R}$ such that $\boldsymbol{\theta} + a\boldsymbol{w}_i + b\boldsymbol{w}_j \in \mathcal{D}$ and differentiable function $f_{\mathcal{K}}:\mathcal{D} \rightarrow \mathbb{R}$,
\begin{equation*}
\begin{split}
    \vert\mathbb{E}[d f_{\mathcal{K}}^2[a\boldsymbol{w}_i + b \boldsymbol{w}_j]]\vert &\overset{(i)}{=} \vert\mathbb{E}[(a d f_{\mathcal{K}}[\boldsymbol{w}_i] + b d f_{\mathcal{K}}[\boldsymbol{w}_j])^2]\vert\\
    & \overset{(ii)}{=} \vert2ab\vert \vert\mathbb{E}[d f_{\mathcal{K}}[\boldsymbol{w}_i] d f_{\mathcal{K}}[\boldsymbol{w}_j]]\vert\\
    & \overset{(iii)}{\leq} \vert2ab\vert \left(\mathbb{E}[d f_{\mathcal{K}}^2[\boldsymbol{w}_i]] \mathbb{E}[d f_{\mathcal{K}}^2[\boldsymbol{w}_j]]\right)^{1/2}\\
    & \overset{(iv)}{=} \vert2ab\vert\sqrt{\lambda_i\lambda_j}\\
    & \overset{(v)}{=} 0.
\end{split}
\end{equation*}

The differentiability of $f_{\mathcal{K}}$ implies continuity so that $\mathbb{E}[d f_{\mathcal{K}}^2[a\boldsymbol{w}_i + b \boldsymbol{w}_j]] =0 \implies d f_{\mathcal{K}}(\boldsymbol{\theta})[a\boldsymbol{w}_i + b \boldsymbol{w}_j] = 0$ for all $\boldsymbol{\theta}$. Equality (i) follows from the linearity of the differential over directions. Equality (ii) is a simplification of the expanded quadratic that follows by assumption---i.e., $\mathbb{E}[df_{\mathcal{K}}^2[\boldsymbol{w}_i]] = \mathbb{E}[df_{\mathcal{K}}^2[\boldsymbol{w}_j]] =\lambda_i = \lambda_j = 0$ for $\boldsymbol{w}_i, \boldsymbol{w}_j \in \text{Null}(\boldsymbol{C}_{\mathcal{K}})$. Inequality (iii) is Cauchy-Schwarz for a Hilbert space of square integrable (measurable) functions. Equality (iv) follows from expression \eqref{eq:eigs} and equality (v) follows from assumption. Finally, by recursively assigning $\boldsymbol{w}_i \longleftarrow a\boldsymbol{w}_i + b\boldsymbol{w}_j$ and taking the next eigenvector from the set to be $\boldsymbol{w}_j$, we can repeat the above for all remaining eigenvectors. The converse ($\impliedby$) follows directly from \eqref{eq:eigs} and linearity of the differential.
\end{proof}

%In fact, either throughput response from \eqref{eq:throughputs} is referred to as a \emph{ridge function} over $\boldsymbol{\theta}$'s if and only if $dS_{\boldsymbol{\theta}}[\boldsymbol{w}] = 0$ for all $\boldsymbol{w} \in \text{Null}(\boldsymbol{W}_r^{\transpose})$. 

Naturally, if the trailing eigenvalues are merely small as opposed to zero, then the function changes much less over the inactive directions with smaller directional derivatives. This lends itself to a framework for reduced-dimension approximation of the function such that we only approximate changes in the function over the first $r$ active directions and take the approximation to be constant over the trailing $D-r$ inactive directions \cite{Constantine2015}. Such an approximation to $f_{\mathcal{K}}$ is called a \emph{ridge approximation} by a function $h_{\mathcal{K}}$ referred to as the \emph{ridge profile} \cite{zahm2020gradient}, i.e.,
\begin{equation}\label{eq:ridge_approx}
    f_{\mathcal{K}}(\boldsymbol{\theta}) \approx h_{\mathcal{K}}(\boldsymbol{W}_{r,\mathcal{K}}^{\transpose} \boldsymbol{\theta}).
\end{equation}
In the event that the trailing eigenvalues of $\boldsymbol{C}_{\mathcal{K}}$ are zero, then the approximation is exact for a particular $h_{\mathcal{K}}$ \cite{Constantine2015}. 

In either case, approximation or an exact ridge profile, the possibility of reducing dimension by projection to fewer, $r< D$, \emph{active coordinates} $\boldsymbol{\gamma} = \boldsymbol{W}_{r,\mathcal{K}}^{\transpose}\boldsymbol{\theta} \in \mathbb{R}^r$ can enable higher-order polynomial approximations for a given data set of coordinate-output pairs and an ability to \emph{visualize} the approximation. For example, we can visualize the approximation by projection to the active coordinates when $r$ is chosen to be $1$ or $2$ based on the decay and gaps in the eigenvalues. These subsequent visualizations are referred to as \emph{shadow plots} \cite{Grey2017} or graphs $\lbrace (\boldsymbol{W}_{r,\mathcal{K}}^{\transpose}\boldsymbol{\theta}_n, f_{\mathcal{K}}(\boldsymbol{\theta}_n))\rbrace_{n=1}^N$ for $N$ samples $\{\boldsymbol{\theta}_n\}$ drawn uniformly. A strong decay leading to a small sum of trailing eigenvalues implies an improved approximation over relatively few important directions while larger gaps in eigenvalues imply an improved approximation to the low-dimensional subspace \cite{Constantine2015}. Identifying if this structure exists depends on the decay and gaps in eigenvalues. We can subsequently exploit any reduced dimensional visualization and approximation to simplify our problem \eqref{eq:multiopt}. However, we must reconcile that our problem of interest involves two separate computations of throughput, $f_{\mathcal{L}}$ and $f_{\mathcal{W}}$. These considerations are addressed in subsection \ref{subsec:mixing}.

\subsection{Subspace Mixing}\label{subsec:mixing}
Independently approximating active subspaces for the objectives $f_{\mathcal{L}}$ and $f_{\mathcal{W}}$ generally results in different subspaces of the shared parameter domain. The next challenge is to define a common subspace that, while not active for each objective, is nevertheless sufficient to capture variability of both KPI simultaneously. Assume that we can reduce important parameter combinations to a common dimension $r$ of potentially distinct subspaces. These subspaces are spanned by the column spaces of $\boldsymbol{W}_{r,\mathcal{L}}$ and $\boldsymbol{W}_{r, \mathcal{W}}$ chosen as the first $r$ eigenvectors resulting from separate approximations of \eqref{eq:C} for LAA and Wi-Fi throughputs, respectively. The challenge is to appropriately ``mix'' the subspaces so we may formulate a solution to \eqref{eq:multiopt} over a common dimension reduction. 

One method to find an appropriate subspace mix is to take the union of both subspaces. However, if $r\geq 2$ and $\text{Range}(\boldsymbol{W}_{r,\mathcal{L}}) \cap \text{Range}(\boldsymbol{W}_{r, \mathcal{W}}) = \lbrace \boldsymbol{0} \rbrace$ then the combined subspace dimension is inflated. This inflation betrays the goals of dimension reduction and, furthermore, hinders prospects for visualization. We use interpolation between the two subspaces to overcome these difficulties and retain the common reduction to an $r$-dimensional subspace. The space of all $r$-dimensional subspaces in $\mathbb{R}^D$ is the $r(D-r)$-dimension Grassmann manifold (Grassmannian\footnote{Formally, an element of the Grassmannian is an equivalence class, $[\boldsymbol{U}_r]$, of all orthogonal matrices whose first $r$ columns span the same subspace as $\boldsymbol{U}_r \in \mathbb{R}^{D\times r}$. That is, the equivalence relation $X \sim Y$ is given by $\text{Range}(X) =\text{Range}(Y)$ denoted $[X]$ or $[Y]$ for $X,Y \in \mathbb{R}^{D\times r}$ full rank with orthonormal columns.}) denoted $\text{Gr}(r,D)$ \cite{edelman1998geometry}. Utilizing the analytic form of a geodesic over the Grassmannian \cite{edelman1998geometry}, we can smoothly interpolate from one subspace to another---an interpolation that is, in general, non-linear. This is particularly useful because the \emph{distance} \cite{edelman1998geometry} between any two subspaces along such a path, 
$
[\boldsymbol{U}_r]:\mathbb{R} \rightarrow \text{Gr}(r, D):s \mapsto [\boldsymbol{U}_r(s)] \,\, \text{for all}\,\, s \in [0,1],
$
minimizes the distance between the two subspaces $\text{Range}(\boldsymbol{W}_{r, \mathcal{L}}), \text{Range}(\boldsymbol{W}_{r, \mathcal{W}}) \in \text{Gr}(r,D)$ defining the geodesic. That is, the geodesic $[\boldsymbol{U}_r(s)]$ minimizes the distance between $\text{Range}(\boldsymbol{W}_{r, \mathcal{L}})$ and $\text{Range}(\boldsymbol{W}_{r, \mathcal{W}})$ while still constituting an $r$-dimensional subspace in $\mathbb{R}^D$.

In an effort to improve the ridge approximations while retaining the ability to visualize the response and trace of the convex quadratic polynomial, we fix $r=2$ and mix the subspaces according to a quadratic approximation with corresponding coefficients of determination $R^2_{\mathcal{L}}$ and $R^2_{\mathcal{W}}$. These are computed over the Grassmannian geodesic using representative subspace coordinates and throughput (coordinate-output) pairs, i.e.,
\begin{equation}
    R^2_{\mathcal{K}}(s) = 1 - \frac{\sum_{n=1}^N(f_{\mathcal{K}}(\boldsymbol{\theta}_n) - h_\mathcal{K}(\boldsymbol{U}_r^{\transpose}(s)\boldsymbol{\theta}_n))^2}{\sum_{n=1}^N(f_{\mathcal{K}}(\boldsymbol{\theta}_n) - 1/N\sum_{n=1}^N f_{\mathcal{K}}(\boldsymbol{\theta}_n))^2}
\end{equation}
for either network represented by $\mathcal{K}$ and quadratic ridge profiles $h_{\mathcal{K}}$. Moreover, the dimension reduction down to $r=2$ is anticipated to help regularize fits $h_{\mathcal{K}}$ and provide more stable approximations resulting from a quadratic trace. We select a criterion to mix subspaces achieving a balanced approximation when $r=2$. This offers the subproblem,
\begin{equation} \label{eq:sub_multi}
    \underset{s \in [0,1]}{\text{maximize}}\,\,\min\lbrace R^2_{\mathcal{L}}(s), R^2_{\mathcal{W}}(s)\rbrace,
\end{equation}
where the separate throughput coefficients of determination are parametrized over a consistent set of subspace coordinates $\boldsymbol{\gamma}_n = \boldsymbol{U}^{\transpose}_2(s)\boldsymbol{\theta}_n$ for all $n=1,\dots,N$ and quadratic ridge profiles $h_{\mathcal{K}}$.

\subsection{Tracing Ridge Profiles}
After approximating $\boldsymbol{W}_{r, \mathcal{L}}$ and $\boldsymbol{W}_{r, \mathcal{W}}$ we must make an informed decision to take the union of subspaces or compute a new subspace $\text{Range}(\boldsymbol{U}_r)$ against some criteria parametrized over the Grassmannian geodesic. Then we may restate the original problem with a common dimension reduction, $\boldsymbol{\gamma} = \boldsymbol{U}_r^{\transpose}\boldsymbol{\theta}$, utilizing updated approximations over $r < D$ mixed active coordinates,
\begin{equation} \label{eq:r_multiopt}
    \underset{\boldsymbol{\gamma} \in \mathcal{Y}}{\text{maximize}} \,\,th_{\mathcal{L}}(\boldsymbol{\gamma}; \boldsymbol{x}) + (1-t)h_{\mathcal{W}}(\boldsymbol{\gamma}; \boldsymbol{x}),
\end{equation}
for all $t \in [0,1]$. Although somewhat intuitive, we formalize stationarity of \eqref{eq:r_multiopt} as a Corollary (dropping the scenario):
\begin{corr} \label{corr:ridge_trace}
If $h_{\mathcal{L}}$ and $h_{\mathcal{W}}$ are ridge functions over a mutual $r$-subspace $\text{Range}(\boldsymbol{U}_r)$, $\boldsymbol{U}_r \in \mathbb{R}^{D \times r}$ such that $\boldsymbol{U}_r^{\transpose}\boldsymbol{U}_r = I_r$, then the necessary conditions of Prop. \ref{prop:ODE} map to subspace coordinates.
\end{corr}
\begin{proof}
By assumption, $f_{\mathcal{K}}(\boldsymbol{\theta}) = h(\boldsymbol{U}_r^{\transpose}\boldsymbol{\theta})$ for either objective $\mathcal{K}$. Consequently, the gradient of $f_{\mathcal{K}}$ with respect to parameters $\boldsymbol{\theta}$ is spanned by the subspace $\text{Range}(\boldsymbol{U}_r)$,
$$
\nabla_{\boldsymbol{\theta}}f_{\mathcal{K}} = \boldsymbol{U}_r\nabla_{\boldsymbol{\gamma}}h_{\mathcal{K}} \implies \boldsymbol{U}_r^{\transpose}\nabla_{\boldsymbol{\theta}}f_{\mathcal{K}} = \nabla_{\boldsymbol{\gamma}}h_{\mathcal{K}}.
$$
Similarly, the Hessian becomes $
\nabla^2_{\boldsymbol{\theta}}f_{\mathcal{K}} = \boldsymbol{U}_r\nabla^2_{\boldsymbol{\gamma}}h_{\mathcal{K}}\boldsymbol{U}_r^{\transpose}$. Rewriting the necessary conditions per Prop. \ref{prop:ODE} such that
$$
\nabla^2_{\boldsymbol{\theta}} \varphi_t(\boldsymbol{\theta}(t)) =\boldsymbol{U}_r\left((1-t)\nabla^2_{\boldsymbol{\gamma}}h_{\mathcal{W}}(t) + t\nabla^2_{\boldsymbol{\gamma}}h_{\mathcal{L}}(t)\right)\boldsymbol{U}_r^{\transpose},
$$
and assigning $\nabla^2_{\boldsymbol{\gamma}} \tilde{\varphi}(t) = (1-t)\nabla^2_{\boldsymbol{\gamma}}h_{\mathcal{W}}(t) + t\nabla^2_{\boldsymbol{\gamma}}h_{\mathcal{L}}(t)$ implies
$$
\boldsymbol{U}_r\nabla^2_{\boldsymbol{\gamma}} \tilde{\varphi}(t)\boldsymbol{U}_r^{\transpose}\dot{\boldsymbol{\theta}}(t) = \nabla_{\boldsymbol{\theta}}f_{\mathcal{W}}(t) - \nabla_{\boldsymbol{\theta}}f_{\mathcal{L}}(t).
$$
For active coordinates $\boldsymbol{\gamma}(t) = \boldsymbol{U}_r^{\transpose}\boldsymbol{\theta}(t)$, the necessary conditions simplify given $\boldsymbol{U}_r^{\transpose}\nabla_{\boldsymbol{\theta}}f_{\mathcal{K}} = \nabla_{\boldsymbol{\gamma}}h_{\mathcal{K}}$,
$$
\nabla^2_{\boldsymbol{\gamma}} \tilde{\varphi}(t)\dot{\boldsymbol{\gamma}}(t) = \nabla_{\boldsymbol{\gamma}}h_{\mathcal{W}}(t) - \nabla_{\boldsymbol{\gamma}}h_{\mathcal{L}}(t),
$$
according to pseudo-time derivative that commutes with matrix multiplication, $\dot{\boldsymbol{\gamma}}(t) = \boldsymbol{U}_r^{\transpose}\dot{\boldsymbol{\theta}}(t)$.
\end{proof}
Once again, this optimization problem involves a closed and bounded feasible domain of parameter values $\mathcal{Y} = \lbrace \boldsymbol{\gamma} \in \mathbb{R}^r \,:\, \boldsymbol{\gamma} = \boldsymbol{U}_r^{\transpose}\boldsymbol{\theta},\,\, \boldsymbol{\theta}\in \mathcal{D}\rbrace$ that remains convex for convex $\mathcal{D}$ and a new subspace $\text{Range}(\boldsymbol{U}_r)$---for hypercube $\mathcal{D}$, $\mathcal{Y}$ is referred to as a \emph{zonotope} \cite{Constantine2015, fukuda2004zonotope}. The utility of the dimension reduction is the ability to formulate a continuous trace of the Pareto front \cite{bolten2020tracing}---involving the inverse of a convex combination of Hessians---in fewer dimensions. In the context of a low-dimensional quadratic trace, the ridge approximations are $h_{\mathcal{K}}(\boldsymbol{\gamma}; \boldsymbol{x}) = \boldsymbol{\gamma}^{\transpose}\boldsymbol{Q}_{\mathcal{K}}\boldsymbol{\gamma} + \boldsymbol{a}_{\mathcal{K}}^{\transpose}\boldsymbol{\gamma} + c_{\mathcal{K}}$ given by $\mathcal{CVX}\left( \{(\boldsymbol{U}_r^{\transpose}\boldsymbol{\theta}_n, f_{\mathcal{K}}(\boldsymbol{\theta}_n; \boldsymbol{x}))\}\right)$ for corresponding network $\mathcal{K}$. Consequently, the resulting active coordinate trace is given by replacing $\boldsymbol{\theta}(t; \boldsymbol{x})$ with the parametrization $\boldsymbol{\gamma}(t; \boldsymbol{x})$ in \eqref{eq:quad_trace} for all $\boldsymbol{\gamma}(t; \boldsymbol{x}) \in \mathcal{Y}$. This is supplemented by visualization in the case $r=1$ or $r=2$ providing \emph{empirical evidence of convexity} and the ability to \emph{visualize the resulting trace}---a powerful tool for facilitating \emph{exploration}, \emph{explainability}, and \emph{comparison}.

However, given the transformation to active coordinates, we are now afforded the flexibility of selecting from an infinite number of inactive coordinates orthogonal to the quadratic trace of \eqref{eq:r_multiopt}. Precisely, restricting approximations to an $r$-dimensional subspace defines a submanifold of approximately Pareto optimal solutions $\mathcal{M} \subseteq \mathbb{R}^D$ given as the product manifold $\mathcal{M} = \boldsymbol{\gamma}([0,1]; \boldsymbol{x}) \times \mathcal{Z}_{\boldsymbol{\gamma}}$ where $\boldsymbol{\gamma}([0,1]; \boldsymbol{x})$ is the image of the quadratic trace over active coordinates and $\mathcal{Z}_{\boldsymbol{\gamma}} = \{\boldsymbol{\zeta} \in \mathbb{R}^{D-r} \,:\, \boldsymbol{U}_r\boldsymbol{\gamma} + \boldsymbol{U}^{\perp}_r\boldsymbol{\zeta} \in \mathcal{D}\}$. This offers a parametrization over a ($D-r+1$)-submanifold of approximately Pareto optimal solutions,
\begin{equation} \label{eq:Pareto_mfld}
    \boldsymbol{\theta}(t, \boldsymbol{\zeta}; \boldsymbol{x}) = \boldsymbol{U}_r\boldsymbol{\gamma}(t; \boldsymbol{x}) + \boldsymbol{U}^{\perp}_r\boldsymbol{\zeta}, 
\end{equation}
for all $\boldsymbol{\zeta} \in \mathcal{Z}_{\boldsymbol{\gamma}(t; \boldsymbol{x})}$. Provided we remain off the boundary of $\mathcal{M}$ such that we are only interested in a trace over the interior of $\mathcal{D}$, we can formulate a \emph{geodesic} over $\mathcal{M}$ of near Pareto optimal solutions, 
\begin{equation} \label{eq:Pareto_geo}
    \boldsymbol{\theta}(t; \boldsymbol{x}) = \boldsymbol{U}_r\boldsymbol{\gamma}(t; \boldsymbol{x}) + \boldsymbol{U}^{\perp}_r\boldsymbol{\zeta}(t),
\end{equation}
where $\boldsymbol{\zeta}(t)$ is a straight-line segment (convex combination) over the Euclidean (zero-curvature) portion of the submanifold $\mathcal{M}$. It is conceivable to parametrize this line segment as the convex combination $\boldsymbol{\zeta}(t) = t\boldsymbol{\zeta}_1 + (1-t)\boldsymbol{\zeta}_0$ where 
$\boldsymbol{\zeta}_0 = \underset{\boldsymbol{\zeta} \in \mathcal{Z}_{\boldsymbol{\gamma}(0; \boldsymbol{x})}}{\text{argmin}}\left(f_{\mathcal{W}}(\boldsymbol{U}_r\boldsymbol{\gamma}(0; \boldsymbol{x}) + \boldsymbol{U}^{\perp}_r\boldsymbol{\zeta}; \boldsymbol{x}) \right)$
and
$\boldsymbol{\zeta}_1 = \underset{\boldsymbol{\zeta} \in \mathcal{Z}_{\boldsymbol{\gamma}(1; \boldsymbol{x})}}{\text{argmin}}\left(f_{\mathcal{L}}(\boldsymbol{U}_r\boldsymbol{\gamma}(1; \boldsymbol{x}) + \boldsymbol{U}^{\perp}_r\boldsymbol{\zeta}; \boldsymbol{x}) \right)$
are the left ($t=0$) and right ($t=1$) inactive solutions, respectively. However, the subsequent inactive geodesic is expected to change throughputs significantly less---resulting in marginal improvements corresponding to inaccuracies induced by a choice of convex quadratic ridge profile.

Lastly, we write the image of throughputs over specific choices of the map $\boldsymbol{\theta}(t; \boldsymbol{x})$ as planar curves through Pareto space (fronts) $\mathcal{P} = \{\boldsymbol{f} \in \mathbb{R}^2 \,:\, \boldsymbol{f} = (f_{\mathcal{W}}, f_{\mathcal{L}})\circ \boldsymbol{\theta}, \, \boldsymbol{\theta} \in \mathcal{D} \}$ containing the true (unknown) Pareto front. In this study, we consider three choices for attempting to approximate the Pareto front: (i) the \emph{geodesic front}
\begin{equation} \label{eq:geo_front}
    %\mathcal{P}_{\boldsymbol{\gamma}(t)} =\\
    \{\boldsymbol{f}(t; \boldsymbol{x}) \in \mathbb{R}^2 \,:\, \boldsymbol{f}(t; \boldsymbol{x})= (f_{\mathcal{W}}, f_{\mathcal{L}} )\circ \boldsymbol{\theta}(t; \boldsymbol{x}) \}
\end{equation}
per \eqref{eq:Pareto_geo}, (ii) the \emph{linear front}
\begin{equation} \label{eq:lin_front}
    %\mathcal{P}_{\boldsymbol{\ell}(t)} = \\
    \{ \boldsymbol{f}(t; \boldsymbol{x}) \in \mathbb{R}^2 \,:\, \boldsymbol{f}(t; \boldsymbol{x}) = (f_{\mathcal{W}} , f_{\mathcal{L}} )\circ \boldsymbol{\ell}(t; \boldsymbol{x})\}
\end{equation}
where $\boldsymbol{\ell}(t; \boldsymbol{x}) = (1-t)\boldsymbol{\theta}_0 + t\boldsymbol{\theta}_1 $ defined by approximated left solution $\boldsymbol{\theta}_0$ (resp. right solution $\boldsymbol{\theta}_1$) to \eqref{eq:multiopt} as a geodesic over Euclidean Pareto manifold, and (iii) the \emph{conditional front}
\begin{multline} \label{eq:cond_front}
    %\mathcal{P}_{\boldsymbol{\zeta} \vert \boldsymbol{\gamma}(t)} = 
    \biggl\lbrace \boldsymbol{f}(t; \boldsymbol{x}) \in \mathbb{R}^2 \,:\, \boldsymbol{f}(t; \boldsymbol{x}) =\\ \int_{\mathcal{Z}_{\boldsymbol{\gamma}(t; \boldsymbol{x})}}(f_{\mathcal{W}}, f_{\mathcal{L}} )\circ \boldsymbol{\theta}(t, \boldsymbol{\zeta}; \boldsymbol{x}) dV(\boldsymbol{\zeta} \vert \boldsymbol{\gamma}(t; \boldsymbol{x}))\biggl\rbrace
\end{multline}
for conditional integral measure over inactive coordinates, $dV(\boldsymbol{\zeta} \vert \boldsymbol{\gamma}(t; \boldsymbol{x}))$. The resulting approximations and visualizations are summarized in Section \ref{sec:Numerics}.

\subsection{Computational Considerations}
In order to approximate the eigenspaces of $\boldsymbol{C}_{\mathcal{L}}$ and $\boldsymbol{C}_{\mathcal{W}}$ for the separate responses \eqref{eq:throughputs} we must first approximate the gradients of the network throughput responses that are not available in an analytic form. Specifically, we use forward finite difference approximations to approximate the partial derivatives in \eqref{eq:C}. These computations are supplemented by a rescaling of all parameters to a unit-less domain that permits consistent finite-difference step sizes.

The rescaling transformation is chosen based on the upper and lower bounds, summarized in Table \ref{Tb1}. This ensures that the scale of any one parameter does not influence  finite difference approximations. Moreover, this alleviates the need for an interpretation or justification when taking linear combinations of parameters with differing units. Because the throughput calculations involve parameter combinations appearing as exponents in the composition of a variety of computations, we use a uniform sampling of log-scaled parameter values. This transforms parameters appearing as exponents to appear as coefficients---a useful transformation given that we ultimately seek an approximation of linear combinations of parameters inherent to the definition of a subspace.

The resulting scaling of the domain is achieved by the composition of transformations
$\boldsymbol{\tilde{\theta}} = \boldsymbol{M}\ln(\boldsymbol{\theta}) + \boldsymbol{b}$
where $ \boldsymbol{M} = \text{diag}(2/(\ln(\boldsymbol{\theta}_{u})_1 - \ln(\boldsymbol{\theta}_{\ell})_1),\dots,2/(\ln(\boldsymbol{\theta}_{u})_D - \ln(\boldsymbol{\theta}_{\ell})_D))$, $\boldsymbol{b} = -\boldsymbol{M}\mathbb{E}[\ln(\boldsymbol{\theta})]$, and $\ln(\cdot)$ is taken component-wise. To compute this transformation, we take $(\boldsymbol{\theta}_{\ell})_i$ and $(\boldsymbol{\theta}_{u})_i$ as the $i$-th entries of the lower and upper bounds and $\mathbb{E}[\ln(\boldsymbol{\theta})]$ the mean of $\ln(\boldsymbol{\theta}) \sim U_D[\ln(\boldsymbol{\theta}_{\ell}), \ln(\boldsymbol{\theta}_u)]$. This particular choice of scaling ensures $\boldsymbol{\tilde{\theta}}\in [-1, 1]^D$ and $\mathbb{E}[\boldsymbol{\tilde{\theta}}] = \boldsymbol{0}$ so the resulting domain is also centered. Lastly, we use Monte Carlo as quadrature to approximate the $\boldsymbol{C}_{\mathcal{K}}$ matrices for both throughputs defined by the integral \eqref{eq:C}, for the two throughputs. The details are provided as Algorithm \ref{alg1}. All approximations, optimizations, and traces mentioned in this work are benefited by these scaling considerations. 

\begin{algorithm}[!ht]
%\captionsetup{font=footnotesize}
\caption{Monte Carlo Approximation of Throughput Active Subspaces using Forward Differences}
\small
    \begin{algorithmic}[1]\label{alg1}
	\REQUIRE Forward maps $f_{\mathcal{L}}$ and $f_{\mathcal{W}}$, small coordinate perturbation $\delta \geq 0$, fixed scenario parameters $\boldsymbol{x}$, parameter bounds $\boldsymbol{\theta}_{\ell}$, $\boldsymbol{\theta}_{u} \in \mathbb{R}^D$, and the number of Monte Carlo samples $N$.
    \STATE Generate $N$ i.i.d. uniform samples, $\lbrace \boldsymbol{\tilde{\theta}}_n\rbrace_{n=1}^N \sim U_D[-1,1]$.
    \STATE Compute $\boldsymbol{M}$, $\boldsymbol{M}^{-1}$, and $\boldsymbol{b}$ according to a uniform distribution of log-scale parameters given $\boldsymbol{\theta}_{\ell}$ and $\boldsymbol{\theta}_{u}$.
    \FOR{$n=1$ to $N$}
    \STATE Transform the uniform log-scale sample to the original scale $\boldsymbol{\theta}_n = \text{exp}(\boldsymbol{M}^{-1}(\boldsymbol{\tilde{\theta}}_n - \boldsymbol{b}))$ where the exponential is taken component-wise.
    \STATE Evaluate forward maps $(f_{\mathcal{L}})_n = f_{\mathcal{L}}(\boldsymbol{\theta}_n; \boldsymbol{x})$ and $(f_{\mathcal{W}})_n = f_{\mathcal{W}}(\boldsymbol{\theta_n}; \boldsymbol{x})$.
    \FOR {$i=1$ to $D$}
    \STATE Transform the $i$-th coordinate perturbation to the original input scale
    $\boldsymbol{\theta}_{\delta} = \text{exp}(\boldsymbol{M}^{-1}(\boldsymbol{\tilde{\theta}}_n + \delta\boldsymbol{e}_i - \boldsymbol{b}))$, where $\boldsymbol{e}_i$ is the $i$-th column of the $D$-by-$D$ identity matrix.
    \STATE Approximate the $i$-th entry of the gradient ($i$-th partial) at the $n$-th sample as
    $$
    (\tilde{\nabla} f_{\mathcal{L}})_{n,i} = \frac{f_{\mathcal{L}}(\boldsymbol{\theta}_{\delta}; \boldsymbol{x}) - (f_{\mathcal{L}})_n}{\delta},
    $$
    similarly for $(\tilde{\nabla} f_{\mathcal{W}})_{n,i}$.
    \ENDFOR
    \ENDFOR
    \STATE Compute the average of the outer product of approximate gradients as
    $$
    \boldsymbol{\tilde{C}}_{\mathcal{L}} = \frac{1}{N}\sum_{n=1}^N (\tilde{\nabla} f_{\mathcal{L}})_{n,:} \otimes (\tilde{\nabla} f_{\mathcal{L}})_{n,:},
    $$
    similarly for $\boldsymbol{\tilde{C}}_{\mathcal{W}}$, where the tensor (outer) product is taken over the $i$-th index replaced by $:$.
    \STATE Compute the eigenvalue decompositions $$\boldsymbol{\tilde{C}}_{\mathcal{L}} = \boldsymbol{\tilde{W}}_{\mathcal{L}}\boldsymbol{\tilde{\Lambda}}_{\mathcal{L}}\boldsymbol{\tilde{W}}_{\mathcal{L}}^{\transpose} \quad \text{and}\quad \boldsymbol{\tilde{C}}_{\mathcal{W}} = \boldsymbol{\tilde{W}}_{\mathcal{W}}\boldsymbol{\tilde{\Lambda}}_{\mathcal{W}}\boldsymbol{\tilde{W}}_{\mathcal{W}}^{\transpose}$$
    ordered by decreasing eigenvalues.
    \STATE Observe the eigenvalue decay and associated gaps to inform a reasonable choice of $r$.
    \RETURN The first $r$ columns of $\boldsymbol{\tilde{W}}_{\mathcal{L}}$ and $\boldsymbol{\tilde{W}}_{\mathcal{W}}$, denoted $\boldsymbol{\tilde{W}}_{r,\mathcal{L}}$ and $\boldsymbol{\tilde{W}}_{r,\mathcal{W}}$.
    \end{algorithmic}
   % \end{minipage}}
\end{algorithm}
\normalsize

The selection of $r$ in Algorithm \ref{alg1} can be automated by, for example, a heuristic that takes the largest gap in eigenvalues \cite{Constantine2015} or the largest gap occurring after thresholding the sum of eigenvalues. For simplicity, we take an exploratory approach to selecting $r$ which requires user-input. We seek a visualization of the response to provide empirical---explainable---evidence that the throughputs are predominantly convex (convex ridge profiles) and hence require $r\leq 2$. We then check that the result offers acceptable approximations of throughputs with sufficient gaps in the second and third eigenvalues suggesting reasonable subspace approximations.

\subsection{Stretch Sampling} \label{subsec:stretch}
Projection of randomly sampled points $\{ \boldsymbol{\theta}_n \}$ to subspace coordinates is conflated by a \emph{concentration of distances} phenomenon. This is demonstrated empirically, for our case, in Fig. \ref{fig:stretch_smpl}. Generally, for otherwise unknown (random) subspace, we anticipate---with high probability---that distances between our random samples will uniformly scale by a factor of $\sqrt{r/D}$ for a random projection to subspace coordinates \cite{EFTEKHARI20111589, baraniuk2009random}. Hence, given an otherwise unknown set of subspaces to be mixed resulting from Algorithm \ref{alg1}, we desire a procedure to systematically improve samples that \emph{may concentrate} over the reduced dimension subspace.

To inform a supplementary space-filling design that produces additional samples for $\mathcal{CVX}\left( \{(\boldsymbol{U}_r^{\transpose}\boldsymbol{\theta}_n, f_{\mathcal{K}}(\boldsymbol{\theta}_n; \boldsymbol{x}))\}\right)$ such that $r\leq 2$, we implement a heuristic procedure referred to as \emph{stretch sampling}. In brief, we systematically sample active coordinates supplemented by random draws over corresponding inactive coordinates, $\mathcal{Z}_{\boldsymbol{\gamma}}$. Given data $\{ \boldsymbol{\theta}_n\}$ and a basis for the mixed subspace $\boldsymbol{U}_r$, we consider the boundary of the convex hull of all projected samples $\{ \boldsymbol{U}_r^{\transpose}\boldsymbol{\theta}_n\}$ and zonotope boundary $\partial \mathcal{Y}$. Utilizing linear interpolation of an ordered (clockwise or counterclockwise) set of vertices from either boundary, we take a chosen number of points to uniformly discretize---e.g., compute $25$ points per boundary via piecewise linear interpolation of the ordered set of vertices over separate boundaries. The uniformly sampled boundary points define a Delaunay triangulation, and we return corresponding Voronoi centers not contained in the interior of the projected-data convex hull (exterior Voronoi centers) as new samples that are \emph{stretched} beyond the extent of the projected data. These new samples---aggregated with the original data---fill out the remainder of the zonotope while inactive coordinates are drawn randomly from $\mathcal{Z}_{\boldsymbol{\gamma}}$ \cite{Constantine2015} at each new sample. The result produces improved sampling to mitigate concerns of extrapolating surrogates over the reduced dimension subspace. A visualization of the new samples using $25$ points per boundary is shown in Fig. \ref{fig:stretch_smpl}.

\begin{figure}[ht] 
    \centering
    \includegraphics[width=0.9\linewidth]{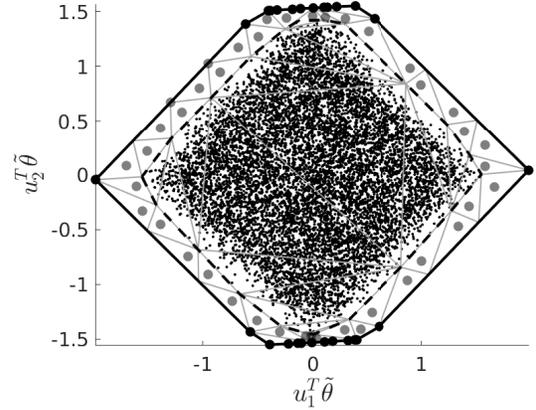}
    \hfill
    \caption{Stretch sampling a two-dimensional zonotope. The boundary of the zonotope $\mathcal{Y}$ (solid black curve) is depicted along with the boundary of the convex hull of data (dashed-black curve) and $N=10,000$ projected random samples (black dots). Axes correspond to independent active coordinates. The resulting triangulation (gray mesh) is informed by $25$ uniformly sampled points per boundary and exterior Voronoi centers (gray dots) constitute the new stretch samples used to improve the ridge approximations.}
    \label{fig:stretch_smpl}
\end{figure}

%% file: Numerics.tex
\begin{figure*}[ht]
\begin{subfigure}{.5\textwidth}
  \centering
  \includegraphics[width=.8\linewidth]{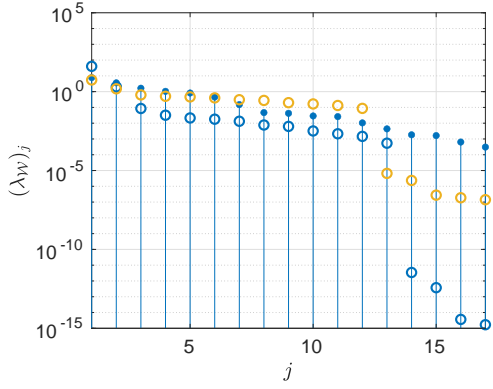}
  \label{fig:wifi_eigs}
  \hfill
  \caption{Wi-Fi Decay of Eigenvalues}
\end{subfigure}
\begin{subfigure}{.5\textwidth}
  \centering
  \includegraphics[width=.8\linewidth]{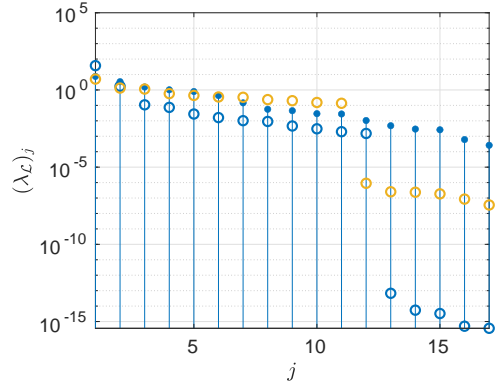}
  \caption{LAA Decay of Eigenvalues}
  \label{fig:laa_eigs}
\end{subfigure}
\caption{Comparison of the eigenvalue decay for Wi-Fi and LAA throughputs. The eigenvalues resulting from Algorithm \ref{alg1} are shown as the blue stem plot (solid blue dots) corresponding to $N=10,000$ Monte Carlo samples. This is contrasted with eigenvalues of the convex quadratic Hessian (yellow circles) computed using the $N=10,000$ unperturbed random function evaluations defining the convex problem \eqref{eq:global_spd_quad} over the full-dimensional parameter space. Additionally, the eigenvalues of the active subspaces resulting from the full-dimensional convex quadratic fit as a surrogate (blue circles) is contrasted to the unbiased finite difference approach (solid blue dots).}
\label{fig:eignevalues}
\end{figure*}
\normalsize

%\begin{figure*}[ht]
%\begin{subfigure}{.5\textwidth}
%    \centering
%    \includegraphics[width=1\linewidth]{./figs/Wifi_activity_score.eps}
%    \label{fig:wifi_activity}
%    \hfill
%    \caption{Wi-Fi Activity Scores}
%\end{subfigure}
%\begin{subfigure}{.5\textwidth}
%    \centering
%    \includegraphics[width=1\linewidth]{./figs/LAA_activity_score.eps}
%    \caption{LAA Activity Scores}
%    \label{fig:laa_activity}
%\end{subfigure}
%\caption{\added[id=ZG]{Parameter activity scores shown with respect to the parameter index (sub-figure left) and the first six parameters with scores ordered from largest to smallest score with corresponding parameter label (sub-figure right). All activity scores are computed according to $r=17$ given the steady decay in eigenvalues. The ordered activity scores also depict the cumulative sum of scores (blue curve) which indicates the first six ordered parameters account for more than $95\%$ of the total sum of scores.}}
%\label{fig:activity_scores}
%\end{figure*}

\begin{figure}[ht]
    \centering
    \includegraphics[width=0.9\linewidth]{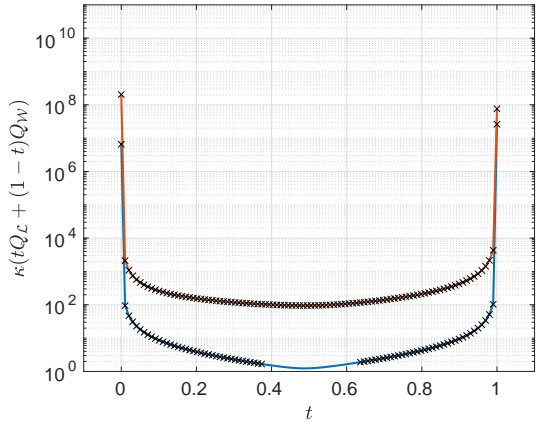}
    \hfill
    \caption{Condition number of convex combinations of quadratic Hessians. The convex combination resulting from \eqref{eq:global_spd_quad} over the full dimensional space (orange) is contrasted with the significantly improved conditioning over the mixed two-dimensional subspace (blue). Results are depicted at $100$ points along the corresponding quadratic traces \eqref{eq:quad_trace}. Black crosses indicate points along the trace that do not pass through the domain.}
    \label{fig:cvx_condition}
\end{figure}

\begin{figure*}[ht]
\begin{subfigure}{.5\textwidth}
  \centering
  \includegraphics[width=0.85\linewidth]{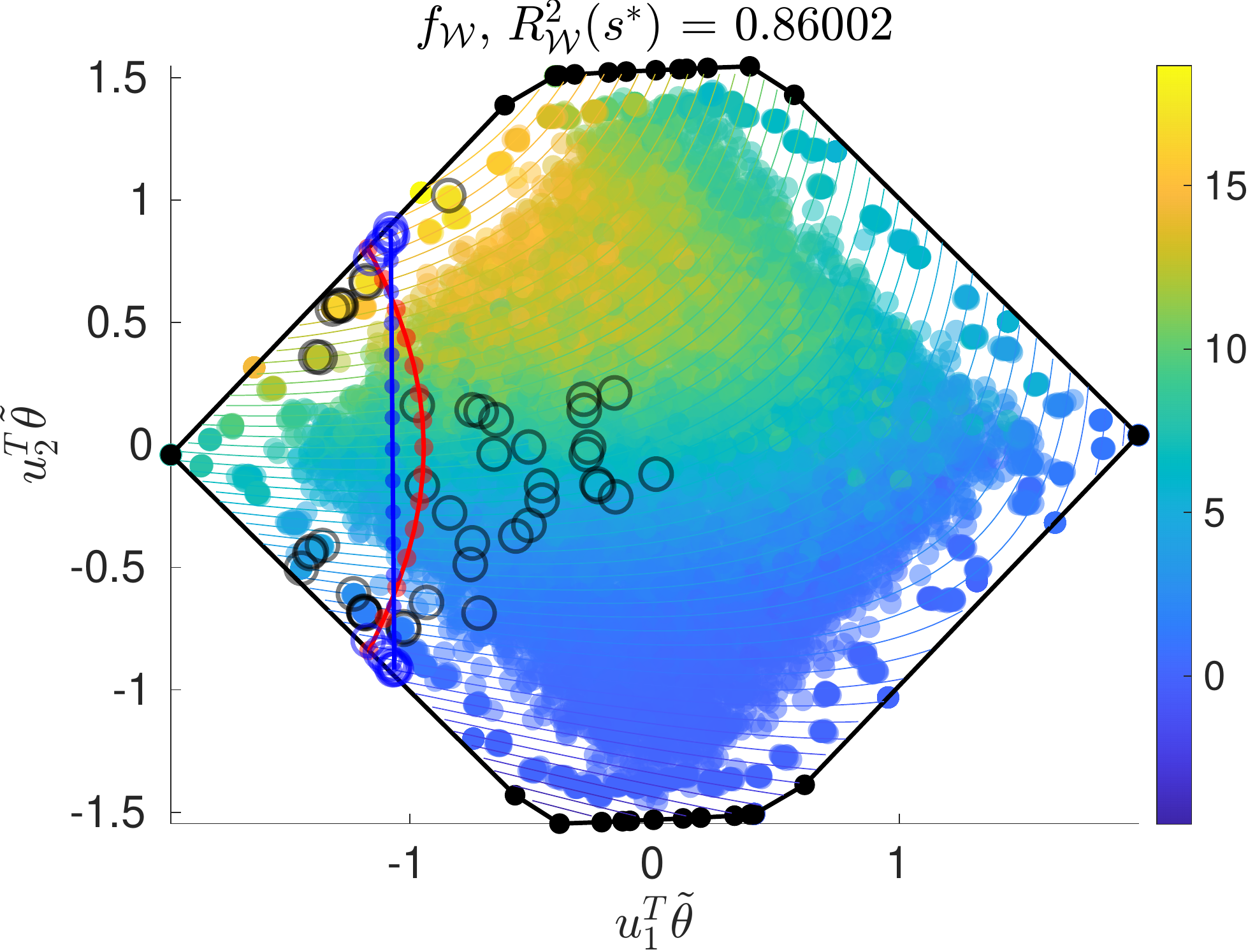}
  \label{fig:sub-first}
  \hfill
  \caption{Wi-Fi Throughput Shadow Plot}
\end{subfigure}
\begin{subfigure}{.5\textwidth}
  \centering
  \includegraphics[width=.85\linewidth]{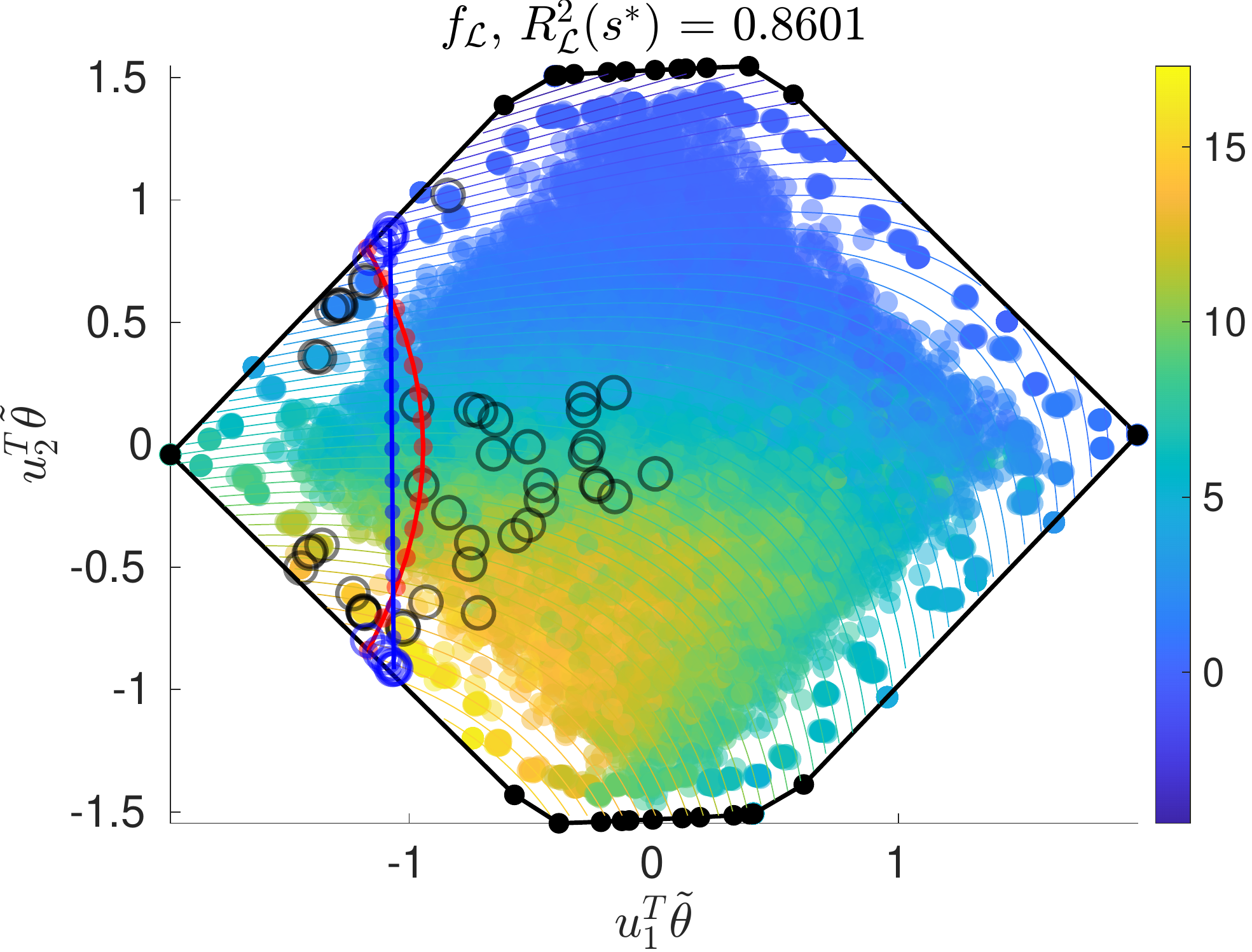}
  \caption{LAA Throughput Shadow Plot}
  \label{fig:sub-second}
\end{subfigure}
\caption{Pareto trace of convex quadratic ridge profiles. The quadratic Pareto trace (red curve and dots) is overlaid on a shadow plot over the mixed coordinates (colored scatter) with the projected bounds and vertices (zonotope) of the domain (black dots and lines), $\mathcal{Y}$. The quadratic approximations (colored contours) are contrasted against the true function evaluations represented by the colors of the scatter. Also depicted is the projection of the non-dominated domain values from the $N=10,000$ random samples (black circles). The active coordinate trace (red dots and curve) begins at the upper left-most boundary with near maximum quadratic Wi-Fi throughput and we move (smoothly) along a trajectory to the lower left-most boundary obtaining near maximum quadratic LAA throughput---maintaining an approximately best trade-off over the entire curve restricted to $\mathcal{Y}$. This is contrasted with a linear interpolation of full dimensional left ($t=0$) and right ($t=1$) approximations to \eqref{eq:multiopt} (blue dots and curve) and $15$ successive approximations to \eqref{eq:multiopt} with uniform discretization of $t\in [0,1]$ (blue circles). }
\label{fig:fig_1}
\end{figure*}
\normalsize

\begin{figure}
\centerline{\includegraphics[width=1\linewidth]{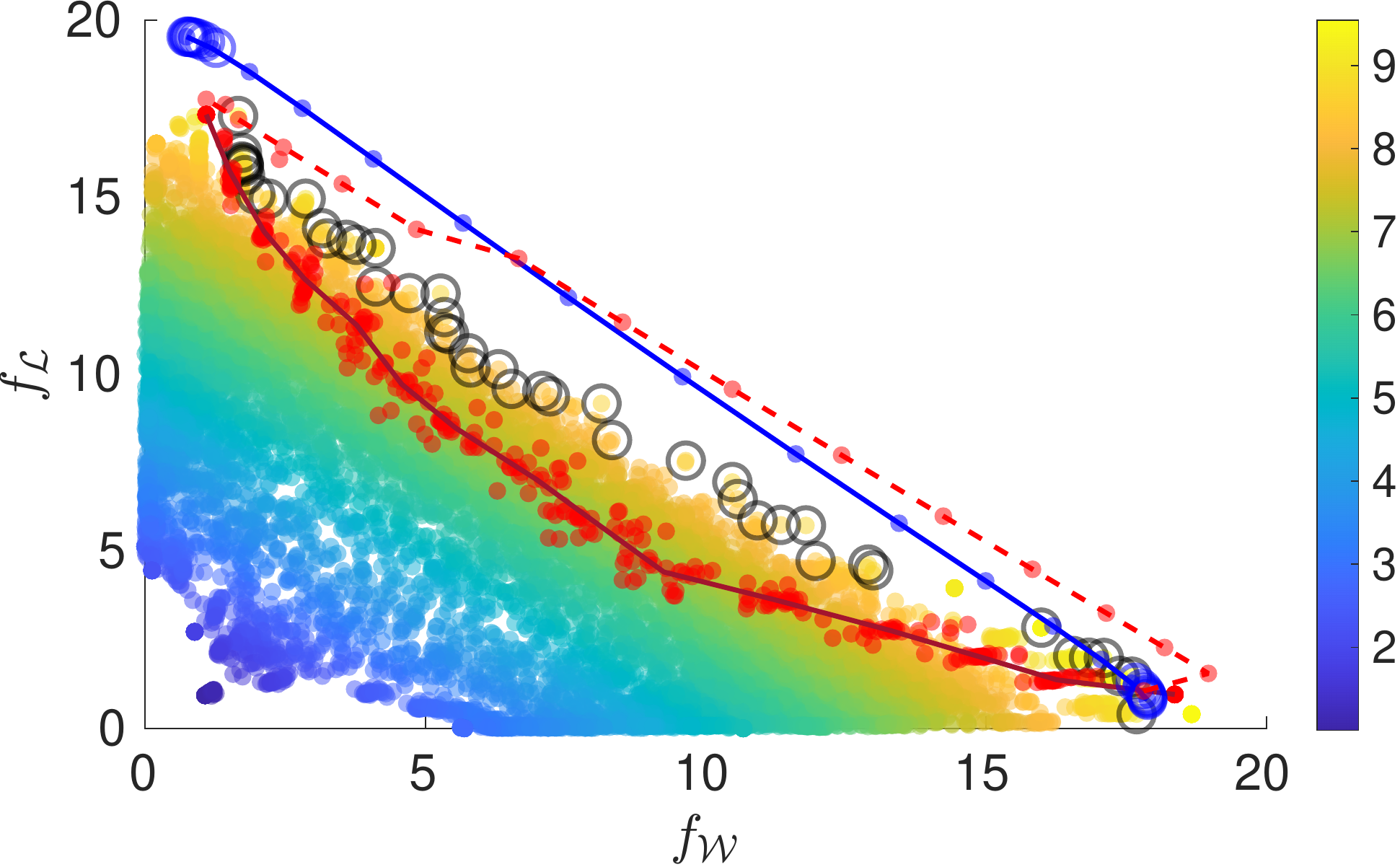}}
\caption{Approximation of the Pareto front resulting from the quadratic trace. The conditional Pareto front (solid-red curve) is contrasted with the non-dominated throughput values (black circles) and scatter of $N = 10,000$ random responses colored according to the averaged throughput ($t=0.5$) response. The solid-red curve is a Monte Carlo approximation of the conditional Pareto front \eqref{eq:cond_front} with corresponding evaluations (overlaid red dots) used to compute conditional means. The dashed-red curve and coincident dots represent a discretization over geodesic \eqref{eq:Pareto_geo} between left ($t=0$) and right ($t=1$) inactive maximizing arguments on the manifold of approximately Pareto optimal solutions \eqref{eq:Pareto_mfld}. Also shown is the linear Pareto front \eqref{eq:lin_front} through the full parameter space (blue curve and coincident dots) and a set of $15$ successive maximizations \eqref{eq:multiopt} over a uniform discretization of $t \in [0,1]$ (blue circles).
}
\label{fig_2}
\end{figure}
\normalsize
\raggedbottom

We demonstrate the ideas proposed in Section \ref{sec:AS} on the LAA-Wi-Fi coexistence scenario described in Section \ref{sec:sys_model} to maximize both throughputs simultaneously per \eqref{eq:multiopt} and subsequent simplification \eqref{eq:r_multiopt}. We apply active subspaces to simplify the multi-criteria optimization problem \eqref{eq:multiopt} by focusing on a reduced set of mixed MAC and PHY layer parameter combinations informing a manifold of near Pareto optimal solutions \eqref{eq:Pareto_mfld}. In particular, we contrast the geodesic trace \eqref{eq:Pareto_geo} with a simple linear interpolation of left and right approximations to \eqref{eq:multiopt} over the full dimensional space---informing \eqref{eq:geo_front} and \eqref{eq:lin_front}, respectively. We also visualize a conditional approximation of the Pareto front by taking the conditional average of throughput responses over inactive coordinates subordinate to $\mathcal{Z}_{\boldsymbol{\gamma}}$ using \eqref{eq:Pareto_mfld}---computed as a Monte Carlo approximation of \eqref{eq:cond_front}. The subsequent visualization of the Monte Carlo approximation of \eqref{eq:cond_front} emphasizes the expected range of throughput values captured, on average, over the near-optimal Pareto manifold \eqref{eq:Pareto_mfld}. Table \ref{Tb1} lists the scenario parameters and parameter bounds (constraints) used by the model for all computations and optimizations. 

The numerical experiment utilizes $N = 10,000$ samples resulting in $N(m+1) = 180,000$ total function evaluations to compute forward differences with $\delta = 10^{-6}$. The eigenvalues resulting from Algorithm \ref{alg1} indicate a steady decay devoid of dramatically different gaps in eigenvalues (see Fig. \ref{fig:eignevalues}). This is contrasted with the resulting decay of eigenvlaues associated with the principal Hessian directions \cite{Li1992} and quadratic active subspace approximations \cite{Grey2017}, both defined with respect to full-dimensional fits of throughputs \eqref{eq:global_spd_quad}, $\mathcal{CVX}\left( \{(\boldsymbol{\theta}_n, f_n)\}\right)$. Interestingly, the dimension reductions informed by the convex quadratic fits admit stronger decays in trailing eigenvalues (in both cases). This indicates a bias resulting from the choice of convex quadratic model over full-dimensional parameter space in contrast to the unbiased estimates of finite differences in Algorithm \ref{alg1} \cite{Constantine2015}. 

Despite the steady decay in eigenvalues resulting from Algortihm \ref{alg1}, the separate active subspaces for corresponding throughputs informed relatively accurate degree-$2$ to degree-$5$ polynomial approximations---computed utilizing sets of coordinate-output pairs $\lbrace (\boldsymbol{\tilde{W}}_{r,\mathcal{L}}^{\transpose}\boldsymbol{\tilde{\theta}}_n, (f_{\mathcal{L}})_n)\rbrace_{n=1}^N$ and $\lbrace (\boldsymbol{\tilde{W}}_{r,\mathcal{W}}^{\transpose}\boldsymbol{\tilde{\theta}}_n, (f_{\mathcal{W}})_n)\rbrace_{n=1}^N$---with varying coefficients of determination between $0.83-0.98$ for both throughputs when $r=1$ or $r=2$. 

The univariate subproblem in \eqref{eq:sub_multi} can be visualized and, in this experiment, achieved a unique maximizing argument $s^* \in [0,1]$ resulting in a mixed subspace with orthonormal basis given by two columns in a matrix $\boldsymbol{U}_r(s^*) = [\boldsymbol{u}_1 \,\, \boldsymbol{u}_2]$ over which separate quadratic ridge functions obtained roughly equal accuracy as approximations to their respective throughputs. The coefficients of determination varied monotonically and intersected over the Grassmannian parametrization. Consequently, the subproblem \eqref{eq:sub_multi} results in an approximately equal criteria for the accuracy of the quadratic ridge profiles $h_{\mathcal{W}}$ and $h_{\mathcal{L}}$, i.e., $R^2_{\mathcal{L}}(s^*) \approx R^2_{\mathcal{W}}(s^*) \approx 0.86$.

Fig \ref{fig:cvx_condition} depicts the condition numbers of convex combinations of quadratic Hessians resulting from a full-dimensional fit \eqref{eq:global_spd_quad}, i.e., $r=D$, and the chosen $(r=2)$-dimensional subspace over the quadratic trace. Examining Fig. \ref{fig:cvx_condition}, we observe that the condition number over the two-dimensional subspace is more than two orders of magnitude lower than the full dimensional fit---achieving near optimal conditioning in the middle of the trace. We also note that at least some portion of the $2$-dimensional trace passes through the domain while the full-dimensional trace does not. Consequently, the implicit regularization over the mixed subspace informed by subproblem \eqref{eq:sub_multi} gives a more stable (and feasible) quadratic trace than simply solving \eqref{eq:global_spd_quad} in full-dimensional space. The convex quadratic ridge approximations, Pareto trace approximations, non-dominated designs from the set of $N=10,000$ random parameters, along with projected random samples and mixed subspace zonotope are shown in Fig. \ref{fig:fig_1}. The Pareto front approximation resulting from various traces are shown with the non-dominated designs in Fig. \ref{fig_2}.

Observing Fig. \ref{fig:fig_1}, we depict traces and approximations over the two-dimensional subspace along with corresponding convex quadratic approximations and the data constituting $N=10,000$ projected samples augmented by stretch sampling. The red curve corresponds to the quadratic trace \eqref{eq:Pareto_geo} while the blue curve corresponds to the linear trace projected to the subspace for the purposes of visualization and comparison. As a ground truth, we depict the collection of projected non-dominated designs (black circles) resulting from sorting the full data set---effectively this may be viewed as a Pareto optimality solution obtained via random grid search. The non-dominated designs are determined from the full set of $N=10,000$ random samples sorted according to \cite{kung1975}. Note, it is not clear through this visualization that the non-dominated designs constitute elements of an alternative continuous approximation of the Pareto front---perhaps represented by an alternative low-dimensional manifold informed by a machine learning procedure. 

In Fig. \ref{fig_2}, we depict the corresponding throughput evaluations from conditional inactive samples over the quadratic trace as red dots along with associated Monte Carlo approximation of the conditional Pareto front \eqref{eq:cond_front}. In other words, the solid-red line (overlapping the cloud of red dots) connects conditional averages of throughputs utilizing $25$ inactive samples along a corresponding discretization of $15$ active coordinates over the subspace-based quadratic trace in Fig. \ref{fig:fig_1}---constituting a Monte Carlo approximation of \eqref{eq:cond_front}. The visualization emphasizes that the throughputs change significantly less over the inactive coordinates in contrast to the range of values observed over the trace. Contrasting the red dots (conditional samples) and solid-red curve (approximated conditional front) to the colored scatter (averaged throughput response) paired with all random evaluations emphasizes that the near-Pareto optimal manifold \eqref{eq:Pareto_mfld} satisfies, on average, an averaged (summed) thoughput, (i.e., $t=0.5$) which is approximately $\geq 7$. The conditional front \eqref{eq:cond_front} moves approximately through the non-dominated designs of a random grid search. Hence, we have supplemented with a near-optimal (predominately flat) Pareto manifold \eqref{eq:Pareto_mfld} which is implicitly regularized as a solution over a low-dimension subspace, and nearly captures the non-dominated designs from a random grid search, on average. 

However, there are infinite $\boldsymbol{\theta}$ in the original parameter space that correspond to points along the quadratic trace with a subset depicted in Fig. \ref{fig:fig_1} (solid-red curve and dots)---i.e., infinitely many $D-r$ inactive coordinate values that may change throughputs albeit significantly less (roughly an additional $15\% - 20\%$) than the two mixed active coordinates, $\gamma_1=\boldsymbol{u}_1^{\transpose}\boldsymbol{\tilde{\theta}}$ and $\gamma_2=\boldsymbol{u}_2^{\transpose}\boldsymbol{\tilde{\theta}}$. To reconcile the choice of infinitely many inactive coordinates, we consider a discretization of $15$ points along the geodesic trace \eqref{eq:Pareto_geo} (curvature and discretization depicted in Fig. \ref{fig:fig_1}) and associated Pareto front approximation \eqref{eq:geo_front} (red-dashed line in Fig. \ref{fig_2}). This is contrasted with $t$ discretized uniformly to approximate $15$ successive optimizations solved in the full dimensional space \eqref{eq:multiopt} (blue circles) and the linear trace \eqref{eq:lin_front} (blue line and dots). There is reasonable agreement in the solutions produced by all three approaches. However, the geodesic trace \eqref{eq:Pareto_geo} produces marginally better solutions from strict Wi-Fi optimization until the interesting intermediate design region (over $t\in [0,0.5]$). The naive approach of aggregating successive optimizations \eqref{eq:multiopt} over uniform discretization of $t$ struggles to identify any intermediate combinations of throughputs with solutions clustering towards one maximum or the other---a recognized issue in multi-criteria problems \cite{das1997closer} remedied by our alternative parameterizations \eqref{eq:geo_front} and \eqref{eq:lin_front}. Interestingly, and unexpectedly, the linear submanifold and subsequent front \eqref{eq:lin_front} perform comparatively well. This may suggest that the curved portion of the Pareto optimal manifold only affords minor improvements and an alternative Pareto manifold of near-optimal solutions could be built from a tubular neighborhood of the line segment interpolating left and right approximations to \eqref{eq:multiopt}.

There is some bias in the approximation of the conditional front (solid-red curve) in Fig. \ref{fig_2} that is not a least-squares curve of non-dominated throughput values (black circles) potentially due in part to the quadratic ridge approximations or regularization by virtue of simplifying over a subspace. However, this issue is reconciled by constructing the geodesic trace \eqref{eq:Pareto_geo} (with corresponding Pareto front as the dashed-red curve in Fig. \ref{fig_2}) that dominates the random grid search (black circles in Fig. \ref{fig_2}) and the majority of the linear trace (blue curve in Fig. \ref{fig_2}). We expect further refinements to these approximations (higher-order polynomials), increased subspace dimension, or numerical integration of the ODE in Prop. \ref{prop:ODE} with simplifications per Cor. \ref{corr:ridge_trace} will further improve the near-optimal Pareto manifold and subsequent traces.
\raggedbottom